\documentclass[12pt,a4paper]{article}

\usepackage[a4paper, left=.5in, right=.5in, top=1in, bottom=1in]{geometry}
\usepackage{amsmath,amssymb,amsfonts}
\usepackage{amsthm}
\usepackage{graphicx}
\usepackage{xcolor}
\usepackage{cleveref}
\usepackage{caption}
\AtBeginDocument{\captionsetup[figure]{font=small}}
\usepackage{subcaption}
\usepackage{braket}
\usepackage{marginnote}
\usepackage{todonotes}
\usepackage{cases}
\usepackage{enumitem}
\usepackage{makecell}
\usepackage{array}
\usepackage{makecell}
\usepackage{authblk}

\newcolumntype{C}[1]{>{\centering\arraybackslash}m{#1}}

\title{Barbell Codes: \\qLDPC Codes for Superconducting Quantum Hardware}

\author[1]{\small Shin Ho Choe\textsuperscript{\textdaggerdbl}\thanks{\texttt{shinho.choe@iqm.tech}}}
\author[1]{Vincent Steffan\textsuperscript{\textdaggerdbl}\thanks{\texttt{vincent.steffan@iqm.tech}}}
\author[1]{Florian Vigneau}
\author[2]{Pedro Parrado-Rodr\'iguez}
\author[1]{Hsiang-Sheng Ku}
\author[1]{Martin Leib}
\author[1]{Francisco Revson Fernandes Pereira}
\author[1]{Fedor \v{S}imkovic IV}

\affil[1]{  IQM Quantum Computers, Georg-Brauchle-Ring 23-25, 80992 Munich, Germany
}
\affil[2]{
  IQM Quantum Computers, P.\ de la Castellana 200, 28046 Madrid, Spain
}

\date{\today}

\newcommand{\nlc}{near-local coupler}
\newcommand{\nlcs}{\nlc s}
\newcommand{\DHex}{six-qubit star lattice}
\newcommand{\move}[0]{\ensuremath{\mathrm{MOVE}}}
\newcommand{\cz}[0]{\ensuremath{\mathrm{CZ}}}

\newcommand{\sR}[0]{\ensuremath{\mathsf{R}}}
\newcommand{\sRX}[0]{\ensuremath{\mathsf{RX}}}

\newcommand{\sCNOT}[0]{\ensuremath{\mathsf{CNOT}}}
\newcommand{\sM}[0]{\ensuremath{\mathsf{M}}}
\newcommand{\sMX}[0]{\ensuremath{\mathsf{MX}}}
\newcommand{\ztype}{\ensuremath{\mathrm{Z}}}
\newcommand{\xtype}{\ensuremath{\mathrm{X}}}

\DeclareMathOperator{\supp}{supp}
\DeclareMathOperator{\modular}{mod}

\newtheorem{theorem}{Theorem}
\newtheorem{lemma}[theorem]{Lemma}
\newtheorem{corollary}[theorem]{Corollary}
\newtheorem{remark}{Remark}

\begin{document}

\maketitle
\begingroup
\renewcommand{\thefootnote}{\textdaggerdbl}
\footnotetext{These authors contributed equally to this work.}
\endgroup

\abstract{The major challenge on the way to fault-tolerant quantum computing comes from the insufficient quality of hardware components and the difficulty of scaling their number without further compromising fidelity. Quantum Low-Density Parity-Check (qLDPC) codes offer a promising solution by encoding logical qubits with low overhead and at a comparatively high code distance. However, it remains an open question how to scalably implement efficient qLDPC codes on fixed-connectivity quantum chips without increasing hardware complexity to enable the non-local interactions in their underlying QEC cycles. We resolve this challenge for the first time by introducing a family of qLDPC ``barbell’' codes accompanied by a realistic chip layout that natively supports all required two-qubit interactions. Crucially, the hardware complexity required to implement barbell codes remains constant as code distance increases. We provide a detailed investigation into the feasibility of all required hardware components and simulate a specific family of barbell codes against circuit-level noise. We find that, with a modest overhead of $<30$ data qubits per logical qubit, barbell codes can preserve information at a physical noise strength of $10^{-4}$ for several trillion QEC cycles. Simulations of logical multi-Pauli measurements, performed with circuits tailored to the chip, yield similar logical performance per QEC round, indicating that entangling gates between logical qubits in barbell codes can be realized fault-tolerantly.}

\vspace{1cm}

Quantum computing raises the hope of efficiently solving problems beyond the reach of classical computers. Recently, substantial progress has been made on finding efficient quantum algorithms for a variety of problems, including cryptography~\cite{webster2026pinnaclearchitecturereducingcost,cain2026shorsalgorithmpossible10000}, data processing~\cite{zhao2026exponentialquantumadvantageprocessing}, and chemistry~\cite{alexeev2025perspectivequantumcomputingapplications}. In order to execute these algorithms, a quantum computer needs to be sufficiently robust against errors, as interesting computations often require the application of billions or even trillions of (non-trivial) quantum gates. All currently known types of quantum hardware have error rates far higher than those required for reliable computations at such a scale. Consequently, the use of quantum error correction (QEC) appears unavoidable for achieving error rates sufficiently low for practical quantum computation.

QEC codes encode $k$ logical qubits into $n$ physical qubits, where $n >k$. This introduces redundancy, which can be used to correct errors that occurred on the way. In order to protect the encoded information in practice, one has to perform so-called QEC cycles implemented by executing two-qubit interactions between the \textit{data qubits} and some additional \textit{syndrome qubits}, which are then measured to obtain the syndrome information required to apply the appropriate correction. One key challenge of using QEC codes efficiently is that the device's connectivity must natively support all two-qubit interactions necessary to read out the syndrome information. 

\begin{figure}
    \centering
    \includegraphics[width=\linewidth]{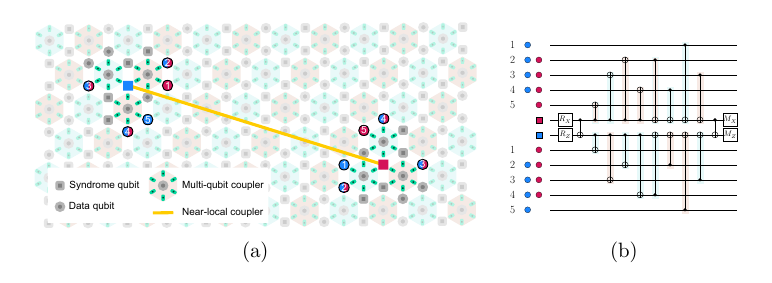}
    \caption{In (a), we give a snapshot of the chip layout for the Barbell architecture. Grey elements are qubits: Square-shaped ones indicate syndrome qubits, circle-shaped ones inside the blue and red hexagons are central elements, and the remaining ones are data qubits. In this way, two-qubit interactions between syndrome and data qubits are realized via central elements. We highlight a pair of \xtype- and \ztype- syndrome qubits together with their supports in red, resp.\ blue. The two syndrome qubits are connected using a \nlc{}. In (b), we depict a possible depth-12 syndrome extraction circuit for the pair of \xtype- and \ztype- stabilizers using superdense syndrome extraction. Using central elements corresponding to hexagons of different colors at each timestep, we ensure that the syndrome extraction circuit can be executed in parallel for all pairs of \xtype- and \ztype- stabilizers.}
    \label{fig:chip_and_superdense}
\end{figure}

Superconducting qubit hardware is a particularly attractive quantum modality thanks to its fast qubit operations: QEC memory experiments below threshold with a syndrome extraction cycle time of the order of $1$ $\mu s$ have been recently demonstrated on hardware~\cite{google2023suppressingerror,google2025belowthreshold,he2025belowthreshold,lacroix2025colorcode}. Similar to all solid-state qubit platforms, they provide a static qubit connectivity graph, chosen at the manufacturing process of the processor. Many different academic groups and companies have realized chips with planar square-grid~\cite{IQMGarnetWhitePaper, arute2019quantum, Krinner_2022} or heavy-hex~\cite{hetenyi2024heavyhex} connectivity, which natively support all two-qubit interactions required for implementing the QEC cycle of the surface code~\cite{McEwen_2023,krinner2022surface17} or the (6.6.6) color code~\cite {gidney2023newcircuitsopensource}. However, scaling up the number of logical qubits encoded into those codes is challenging due to their low encoding rate: a surface code architecture capable of reliably running real-world applications will have an overhead of hundreds to thousands of physical qubits per logical qubit. This renders the corresponding manufacturing processes extremely challenging from both the engineering and economic perspectives.

This incited the community to consider a more general family of QEC codes, called \textit{quantum Low-Density Parity-Check (qLDPC) codes}~\cite{Breuckmann_2021}. For all qubits in the qLDPC code family, each physical qubit is contained in at most a constant number of stabilizer measurements, and each stabilizer acts on at most a constant number of qubits. In order to realize the QEC cycle of any QEC code, one needs to perform two-qubit gates between a syndrome qubit and the data qubits contained in the corresponding stabilizer check. For qLDPC codes, one typically needs two-qubit gates between non-nearest-neighbor qubits that can be located far apart on the processor. Small examples of qLDPC codes have been run on quantum hardware previously. Specifically, a 32-qubit bivariate bicycle (BB) code has been implemented on a superconducting chip using \nlcs{} located on the same layer as the qubits~\cite{wang2026demonstrationofqldpc}.  For this, air bridges~\cite{chen2014fabrication,song2019generation} have been used to cross \nlcs{}, which is expected to lead to inferior hardware performance as the system size grows, particularly due to crosstalk~\cite{bu2025tantalum_airbridges}. 

Another previously proposed approach is the vertical integration of multiple chips to route \nlcs{} through many hardware layers. This approach can be facilitated thanks to recent developments in bump bonds~\cite{mit20173d_3dintegration, RigettiPhysRevApplied2024, vtt2024flipchipprocessor, wallraff2026interchipcouplers}, through-silicon vias (TSV) and interposers~\cite{mit2020tsv,yost2020solid}, as well as on-chip long-distance couplers~\cite{MarxerLongDistanceTransmonCoupler2023,heya2025randomized,xu2026tunable2026}. Despite this progress, the hardware complexity of routing numerous \nlcs{} remains a significant challenge for quantum processor manufacturing, even with sophisticated routing strategies~\cite{zhao2025simpleuniversalrouting, zhou2025louvre, mathews2026placing}. One particular difficulty is that \nlcs{} must be routed through multiple chip layers to avoid having too many air bridges~\cite{mathews2026placing}. 

In this article, we introduce \emph{barbell codes}, a highly efficient qLDPC code family that can be implemented in a realistic near-term superconducting QPU with low hardware complexity. Crucially, we present them together with a designated chip layout and a syndrome extraction cycle consisting solely of two-qubit interactions native to that layout. With those advances, we pave the way for using qLDPC codes on superconducting hardware, yielding qubit-overhead savings of up to a factor of 8 compared to the usual rotated surface code. We validate our approach by simulating a family of highly efficient barbell codes in the presence of circuit-level noise, demonstrating performance comparable to that of a surface code with the same distance. 

We also compute logical error rates per round for logical Pauli measurements for a distance-8 barbell code using the protocol introduced in~\cite{yang2025planarfaulttolerantquantumcomputation}. Here, we observe that the logical error rate is only slightly higher than that of a memory experiment, which strongly suggests that barbell codes are a good candidate for fault-tolerant quantum computation on near-term hardware. 

We now briefly describe the chip layout; for more details on the hardware components, we refer to the Supplementary Information. We propose a flip-chip layout with two layers of connectivity. The first layer hosts the qubits and local couplers while the second layer contains \nlcs{} that enable non-planar connectivity between non-adjacent qubits. Our chip layout is composed exclusively of components whose functionality has been experimentally demonstrated. We term our layout as the ``\DHex{} plus \nlc{} (6QSL+NLC) architecture'' or simply the ``Barbell architecture''. A visual representation of this architecture can be found in \Cref{fig:chip_and_superdense} (a).

\begin{figure}[t]
    \centering

    \begin{subfigure}[t]{0.48\textwidth}
        \centering
        \begin{minipage}[c][0.35\textheight][c]{\linewidth}
            \centering
            \includegraphics[width=\linewidth]{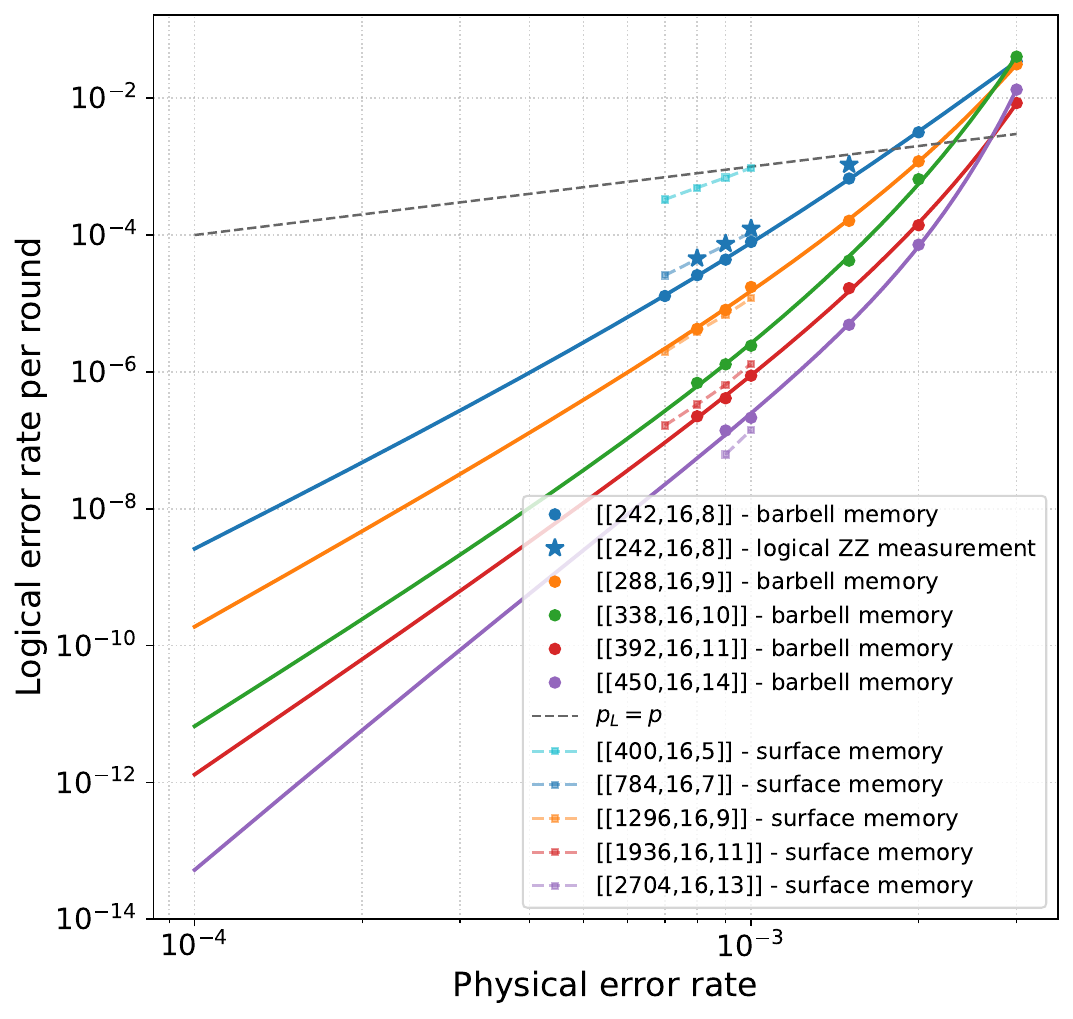}
        \end{minipage}
        \caption{}
        \label{fig:superdense}
    \end{subfigure}
    \hfill
    \begin{subfigure}[t]{0.48\textwidth}
    \small
        \centering
        \begin{minipage}[c][0.35\textheight][c]{\linewidth}
            \centering
            \begin{tabular}{l|c|c}
            \makecell{Experiment, \\Parameters}
                 & \makecell{LER \\($p = 0.0009$)}

                 & \makecell{LER \\($p = 0.001$)}\\
                \hline \hline
                
                M, $[[242,16,8]]$  & $4.4\times 10^{-5}$ & $7.9 \times 10^{-5}$ \\
                \hline
                
                M, $[[288,16,9]]$  & $8.0\times 10^{-6}$ & $1.7\times 10^{-5}$ \\
                \hline
                
                M, $[[338,16,10]]$ & $1.3\times 10^{-6}$ & $2.4\times 10^{-6}$ \\
                \hline
                
                M, $[[392,16,11]]$ & $4.1\times 10^{-7}$ & $8.8\times 10^{-7}$ \\
                \hline
                
                M, $[[450,16,14]]$ & $1.4\times 10^{-7}$ & $2.1\times 10^{-7}$ \\
                \hline
                
                L, $[[242,16,8]]$  & $7.4\times 10^{-5}$ & $1.2\times 10^{-4}$ \\
                \hline
            \end{tabular}
        \end{minipage}
        \caption{}
        \label{fig:table}
    \end{subfigure}

    \caption{Performance of a family of barbell codes under uniform circuit-level noise. All codes belong to the same family, more precisely, are barbell codes of different sizes constructed from the barbell shown in \Cref{fig:chip_and_superdense}~(a). In (a), we show performance curves for barbell codes under circuit-level noise for memory experiments (circles), logical measurement (stars), and, as a comparison, 16 patches of surface code (squares). All logical error rates are plotted per round, and estimated from simulating $d$ rounds with at least 100 failures sampled. In (b), we highlight some of the logical error rates for $p = 0.0009$ and $p = 0.001$. Here, the first column specifies the experiment type (M for memory and L for logical measurement) and the code parameters.}
    \label{fig:LERcurvesandtable}
\end{figure}

In the first layer, qubits and local couplers are organized according to the \DHex{} architecture~\cite{vigneau2025quantumerrordetectionstararchitecture}. The \DHex{} architecture consists of hexagonal cells of six qubits that form a honeycomb lattice. Inside each cell, the six qubits are coupled by a multi-mode coupler that consists of a central element connected to each of the six qubits by a tunable coupler in a star topology. Any pair of qubits within the same cell can perform a two-qubit interaction either by swapping one of the qubits to the central element before and then applying the desired interaction with the second qubit~\cite{renger2025superconductingqubitresonatorquantumprocessor}, or through effective multi-mode interaction~\cite{orell2026efficientandaccuratetwoqubitgate}. In this way, every qubit can perform two-qubit interactions with any of the 12 qubits it shares a central element or a multi-mode coupler with. Parallel two-qubit interactions are possible within the lattice as long as multi-mode couplers are involved in at most one of the two-qubit interactions at a time (i.e., one interaction per hexagonal cell). It is possible to have all qubits involved in two-qubit interactions at the same time.

The second layer hosts the \nlcs{} that connect pairs of qubits that do not share a multi-qubit coupler in the \DHex{}. \nlcs{} can be routed in a other chip layers using existing 3D integration technology~\cite{RigettiPhysRevApplied2024}. Thanks to the translational invariance of the \textit{tile codes}, all the \nlcs{} are parallel and of the same length, which enables the containment of all {\nlcs} to a single layer and significantly simplifies the routing in the chip. A calculation of the hardware complexity according to the methods of \cite{mathews2026placing} is available in the Supplementary Information. We stress that for a fixed family of barbell codes, the length of \nlcs{} is fixed and independent of the code's distance.

Let us now describe the barbell codes along with their QEC cycles. Barbell codes are a variant of \textit{tile codes} introduced in \cite{steffan2025tile, liang2025planarldpccodes}, which are a family of translationally invariant, 2D local qLDPC codes that can be constructed with flexible locality constraints. We review the general construction of tile codes in Methods. Here, we only mention that tile codes are constructed from a pair of $X$- and $Z$-tiles. Translations of these tiles across the lattice make up the set of stabilizers. At the boundary, some stabilizer tiles may be truncated.

Barbell codes are a special case of tile codes in which the stabilizers are tailored to measurement on a \DHex{} plus \nlc{} architecture. Concretely, the \nlcs{} are used exclusively to connect pairs of $X$ and $Z$ stabilizers. In this way, the syndrome information can be obtained using \textit{superdense syndrome extraction circuits}~\cite{gidney2023newcircuitsopensource}, as illustrated in \Cref{fig:chip_and_superdense}~(b). In the first step of the QEC cycle, the pair of $X$- and $Z$-type syndrome qubits is entangled using a two-qubit gate realized by the \nlc{}. In this way, both involved syndrome qubits can be used to collect syndrome information from data qubits for both associated stabilizers. Therefore, the stabilizers that can be measured natively in this architecture are contained in the union of the two neighborhoods of the two syndrome qubits. An example of such stabilizers is shown in \Cref{fig:chip_and_superdense}~(a). 

Let us now discuss our simulation setup for barbell codes. We choose to simulate barbell codes with weight-8 stabilizer generators depicted in \Cref{fig:chip_and_superdense}~(a) that exhibit particularly efficient code parameters. 
As our noise model, we choose uniform depolarizing circuit-level noise.

We investigate the performance of barbell codes under circuit-level noise for both memory and logical computation experiments: In memory experiments, a logical $\ket{\overline{0}}^{\otimes 16}$, resp. $\ket{\overline{+}}^{\otimes 16}$, state in a distance $d$ barbell code is initialized, followed by $d$ rounds of syndrome extraction using the superdense syndrome extraction circuit. Finally, the physical qubits are being measured in the corresponding basis. The collected syndrome information has been decoded using Relay-BP~\cite{mullerImprovedBeliefPropagation2025}. For real-world applications of quantum computers, one needs to be able to perform complex calculations using quantum circuits on the encoded quantum information. We demonstrate that barbell codes can do so efficiently by simulating a multi-qubit Pauli measurement on the encoded qubits in the presence of circuit-level noise. For that, we use the protocol introduced in~\cite{yang2025planarfaulttolerantquantumcomputation}, where the authors showed that logical multi-qubit Pauli measurements, both within a single patch of a tile code and between patches of tile codes, can be performed without additional connectivity requirements. Note that, considering that barbell codes belong to the tile code family, we can directly use this protocol. To measure logical Pauli operators within a patch, one needs to extend the patch horizontally, resp.\ vertically, depending on whether one wants to measure an $X$- or $Z$-type Pauli operator. To perform a joint logical Pauli measurement between two patches of tile codes, one needs to use an ancillary region between the two patches and apply a merge and a split phase, similar to standard lattice surgery. Crucially, the connectivity of the stabilizers of the two patches, merged via the ancillary system, is the same as that of the two individual tile codes. The parity of the stabilizers in the ancillary region determines the outcome of the logical Pauli measurement.

Using standard techniques from the Pauli-based computation model~\cite{paulibasedcomputation}, one can show that joint Pauli measurements are sufficient to perform any CNOT between logical qubits. It is therefore enough to verify that these joint Pauli measurement operations are fault-tolerant. We describe the technical details of this protocol in the Methods section. To ensure that joint Pauli measurements are fault-tolerant, we conduct two types of experiments for each measurement within one patch of the barbell code. Since the setup is symmetric, we restrict ourselves to measuring logical $Z$ operators. In one experimental setup, we initialize the data qubits of the codes in the $Z$-basis. Performing $d$ rounds of stabilizer measurements results in the qubits being in a logical all-$\ket{\overline{0}}$ state. We then apply a circuit measuring one of the logical qubits in the $Z$-basis. This should leave all other logical qubits untouched; the outcome of this measurement should be $+1$, and the logical qubit itself should be in the logical $\ket{\overline{0}}$ state. By measuring all physical qubits in the $Z$ basis at the end of the experiment, we verify that they are all in the $\ket{\overline{0}}$-state. We also include the product of ancillary stabilizers, making up the result of the logical measurement, as a condition for the experiment to succeed. In a second experiment, we initialize all data qubits of the code in the $X$-basis, resulting -- after the stabilizer measurement -- in an encoded all-$\ket{\overline{+}}$-state. In this setup, the measurement result in the $Z$ basis is random. The unmeasured logical qubits, however, should not be affected by this operation. We thus verify whether this experiment has succeeded by measuring all data qubits at the end of the circuit in the $X$-basis and checking whether the unmeasured qubits are being affected. 

In \Cref{fig:LERcurvesandtable}, we present our main findings. Remarkably, we reach the teraquop regime with a distance-14 barbell code with a physical error rate above $10^{-4}$, enabling several trillion QEC rounds. We observe that, for physical error rates at most $10^{-3}$, barbell codes of distance $d$ are on par with surface code patches of comparable distance, while saving up to a factor of $7.0$ in qubit overhead. For example, in \Cref{fig:LERcurvesandtable} (a), we can see that, using 400 data qubits to encode 16 logical qubits at a realistic near-term error rate of $10^{-3}$, one could either implement $16$ patches of distance-5 surface code or one patch of a distance-11 barbell code: In this set-up the logical error rate using a barbell code is nearly three orders of magnitude lower than the one of surface codes. More precisely, while a distance-5 surface code yields a logical error rate per round of $9.6\times 10^{-4}$, a distance-11 barbell code utilizing the same number of data qubits yields a logical error rate per round of $8.8\times 10^{-7}$. This showcases that for near-term hardware with a limited availability of physical qubits, barbell codes already offer a drastic improvement over surface codes. We also observe that the per-round logical error rate for the distance-8 logical measurement is only slightly larger than in the corresponding memory experiment, demonstrating that fault-tolerant quantum computation is possible for barbell codes. It is important to note that we also investigated barbell codes with a stabilizer weight of 10 (see Supplementary Information for details).  While we found that codes with higher stabilizer weight exhibit more attractive code parameters (more precisely, an improvement in code efficiency~$kd^2/n$ up to $8$ compared to  surface code patches), their QEC cycle naturally has a depth surplus of $2$. While this currently degrades performance in circuit-level simulations, further improvements in decoding, for example, may increase the relative attractiveness of the weight-10 barbell codes in the future.

While the connectivity requirements of the six-qubit star lattice are somewhat more demanding than those of the square lattice, they are outweighed by the significant benefits the barbell code brings for actual large-scale quantum hardware. First of all, as we mentioned before, the barbell codes we investigated achieve a physical qubit-per-logical-qubit overhead of up to 8 times lower than a surface code of the same code distance. Thanks to the \nlcs{} being parallel, they can all be routed within a single hardware layer without air-bridges. This is reflected in a hardware complexity metric of only $C_{\mathrm{hw}} \approx 1.65$, calculated in the Supplementary Information following the framework of Ref.~\cite{mathews2026placing}, which is below that of other known tile codes for comparable encoding rates. 
Moreover, the length of all \nlcs{} is fixed and identical. Scaling up the distance of the code or performing lattice surgery between barbell codes requires the exact same chip layout and coupler length as a small memory experiment. 
Also, the \DHex{} only needs to connect three or four couplers per physical qubit (against six for weight-6 BB codes), which relaxes design constraints.

The closest comparison to barbell codes found in the literature is the bivariate bicycle (BB) codes~\cite{linQuantumTwoblockGroup2023,bravyi2023highthreshold,eberhardt2024pruningqldpccodesbivariate,eberhardt2024logicaloperatorsfoldtransversalgates,chen2025anyontheorytopologicalfrustration, liang2025generalizedtoriccodestwisted, liang2024operatoralgebraalgorithmicconstruction}. BB codes require multiple layers to route all the \nlcs{}, which can complicate the chip fabrication process. For instance, according to Ref.~\cite{mathews2026placing},  the hardware complexity metric of $[[144,12,12]]$ BB code~\cite{bravyi2023highthreshold} exceeds 3. The length of the longest coupler scales with the code's distance since these codes have periodic boundaries. This forces one to construct the QPU hardware tailored to a specific BB code and distance. Moreover, modern techniques for implementing logical operations also require different connectivity than the QEC cycle of the codes themselves~\cite{he2025extractorsqldpcarchitecturesefficient,swaroop2025universaladaptersquantumldpc,cross2025improvedqldpcsurgerylogical,yoder2025tourgrossmodularquantum}, making the corresponding hardware difficult to engineer. 

In summary, we presented the first gap-free protocol for implementing qLDPC codes, specifically barbell codes, on superconducting quantum hardware. The combination of multi-qubit couplers, \nlcs{}, and superdense syndrome extraction enables the construction of high-weight and non-local stabilizers while keeping hardware complexity low. We have simulated the performance of these codes under circuit-level noise, both for memory and logical Pauli measurements. Our codes show comparable performance under circuit-level noise to the surface code and achieve up to an 8-fold improvement in code-parameter efficiency. Our work demonstrates that qLDPC codes are feasible on superconducting hardware and reports significant overhead savings for both mid-sized and large-scale quantum processors.

This work gives rise to multiple open questions: (1) We expect that by using a custom-made decoder for barbell codes, the performance and run-time can be improved further. We leave the task of developing such a decoder to future work. (2) In \cite{yang2025planarfaulttolerantquantumcomputation}, the universal fault-tolerant computation for tile codes was demonstrated by coupling tile codes to surface codes. It remains an open question whether this is viable -- to be verified by circuit-level noise simulation. (3) As an alternative to (2), one may ask which kind of gates, such as $H$, $S$, and $T$ gates, one can execute natively on tile codes, or, more specifically, barbell codes. Any future improvements related to these points could yield even greater overhead savings.

\bibliographystyle{alpha}
\bibliography{barbell}

\newpage
\section*{Methods}
\subsection*{Tile codes}
We start by explaining the construction of tile codes~\cite{steffan2025tile, breuckmann2025logicaloperatorsderivedautomorphisms, liang2025planarldpccodes}; for a visualization of the process, see \Cref{fig:tile-codes}. For a tile code, the data qubits sit on the edges of a regular 2D grid. We say a data qubit is \emph{horizontal} or \emph{vertical} if it is located at a horizontal or a vertical edge. Choose a pair of $X$- and $Z$-type stabilizer tiles confined in boxes of size $(D+1) \times (D+1)$.
See \Cref{fig:tile-codes} (a) for a concrete example of stabilizer tiles where $D = 3$. Each vertex and plaquette in the square grid may host one $X$- and $Z$-type stabilizer tile, respectively.

\begin{figure}
    \centering
    \includegraphics[width=\linewidth]{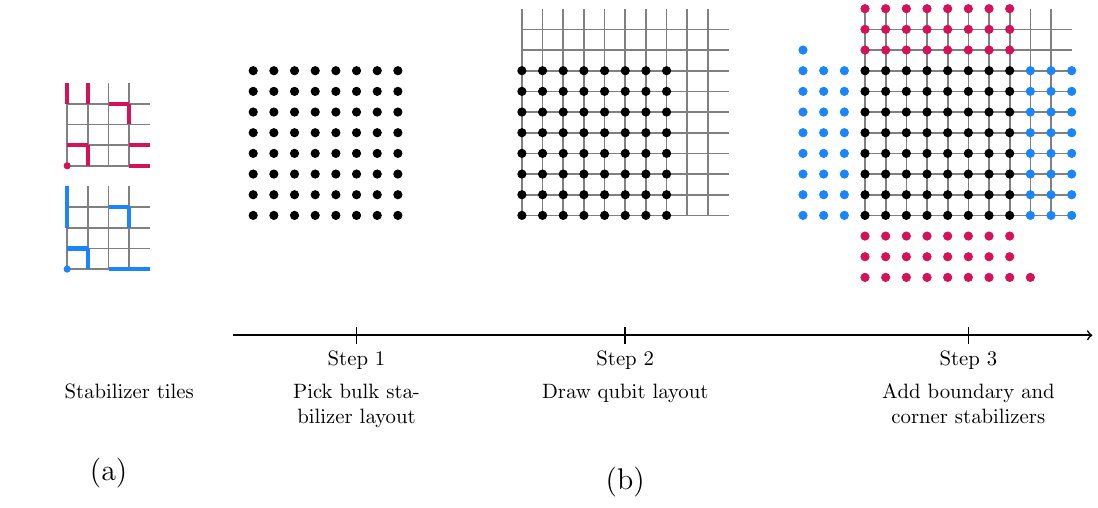}
    \caption{Construction of a tile code. In (a), we show a specific choice of $X$- and $Z$-type tiles confined in boxes of size $(D+1)\times (D+1)$, here with $D=3$. In (b), we depict the process of constructing a tile code: One starts by drawing a lattice of bulk stabilizers of size $(L-D) \times (M-D)$. Here, we set $L = M = 11$. The set of qubits is the union of all boxes at these bulk stabilizers. Then, one chooses $D$ rows resp.\ columns of $X$- resp.\ $Z$-boundary stabilizers as well as corner stabilizers if possible. Later, we will depict tile codes by shifting the $Z$-tile by half-integer values. To avoid confusion, we will depict them with squares instead of circles and put circles on the edges marking data qubits.}
    \label{fig:tile-codes}
\end{figure}

The construction of a tile code proceeds in four steps. We visualize the process in \Cref{fig:tile-codes} (b).
In \textbf{Step 1}, we place bulk stabilizers in the rectangular region of the size $(L-D) \times (M-D)$. In \textbf{Step 2}, we obtain a bigger rectangular region of size $L \times M$ in the top and right directions by $D$. Each data qubit sits on the edges of this rectangular region.
Note that by construction, bulk stabilizers of the same type have the same support up to translation. In \textbf{Step 3}, we add boundary stabilizers.
The support of a boundary stabilizer at a vertex or a plaquette is defined as the support of the stabilizer tile translated to the vertex or the plaquette restricted to the $L \times M$ region.
We add $D$ layers of $X$-type boundary stabilizers on the top and the bottom of the bulk stabilizers. We also add $D$ layers of $Z$-type boundary stabilizers on the left and the right sides of the bulk stabilizers.
In addition, we add boundary stabilizers at the four corners if they commute with already existing stabilizers. In \textbf{Step 4}, we remove qubits that are not supported by either $X$-type or $Z$-type stabilizers. After removing qubits, we also remove a stabilizer if its support is empty.
Note that qubits are not removed in all examples of this work.

\subsection*{Lattice surgery for tile codes}

We will now briefly describe the protocol from~\cite{yang2025planarfaulttolerantquantumcomputation} performing targeted fault-tolerant logical measurements on tile codes. We will explain the construction for measurements of logical $Z$-type Pauli operators in a single patch; the measurement of $X$-type operators is analogous. 

It has been shown that any logical operator of a tile code is supported near the boundary and is always a product of ``omitted stabilizers''; see \Cref{fig:logicalmeasurement} for an example of one such logical operator. Take a logical $Z$-type operator $L_Z$, as for example depicted in \Cref{fig:logicalmeasurement}, and fix some extension length $E$. Then one can construct an extended version of the tile code in which $L_Z$ is the product of the additional $Z$-type stabilizers. 
We add additional ancilla qubits on the supports of the additional $Z$-type stabilizers. We also add additional $X$-type stabilizers on the extended region if they commute with all $Z$-type stabilizers, including the new ones. We mention that a similar protocol is available for measuring products of Pauli operators distributed between two patches of tile codes.

\begin{figure}
    \centering
    \includegraphics[width=0.8\linewidth]{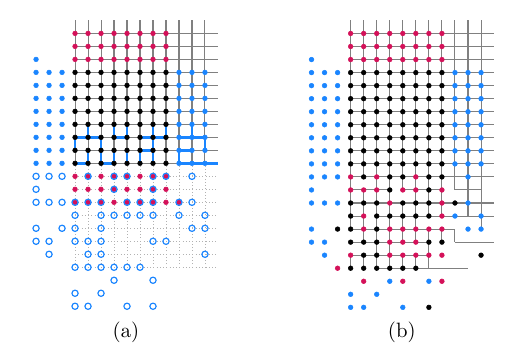}
    \caption{A tile code and its extension for measuring the logical Pauli $Z$-operator corresponding to the logical operator depicted in (a). In an extended code, this logical operator is the product of omitted Pauli $Z$ stabilizers highlighted in (a) as open circles. Extending the code and including ancilla qubits on the edges and stabilizers as shown in (b) lets one measure this logical operator fault-tolerantly.}
    \label{fig:logicalmeasurement}
\end{figure}

\subsection*{Barbell codes on the \DHex{} plus \nlc{} architecture}
Here, we construct a specific example of weight-8 barbell codes compatible with the Barbell architecture.  
The construction is based on the tile code specified by the rectangular region of size $L \times M$ and the stabilizer shown in \Cref{fig:tile-codes}. 
The formal definition of barbell codes and other examples of barbell codes are provided in the Supplementary Information. For a visualization of the process, see \Cref{fig:main-figure}.

In the first step, we add $X$-type ($Z$-type) ``dummy stabilizers'' with empty supports at all vertices (plaquettes) where corresponding plaquettes (vertices) are not occupied with $Z$-type ($X$-type) boundary stabilizers.

In the second step, we place ancilla qubits at the vertices and on the plaquettes occupied by stabilizers, including dummies. 
Ancilla qubits on the vertices (plaquettes) of the coordinate are called the $X$-check ($Z$-check) qubits.

In the third step, we may translate each type of qubit by a certain amount, if necessary. 
Here, we translate all $Z$-check qubits, horizontal data qubits, and vertical data qubits by $(-2,2)$, $(-3, 0)$, and $(-2, -1)$, respectively. 
Note that this translation is also applied to the stabilizer supports. We visualize the translated tile code in \Cref{fig:main-figure} (b).

As a final step, we place qubits on top of the Barbell architecture. 
The architecture consists of two types of hexagons, named \emph{type A} and \emph{type B}, together with \nlcs{}. 
Precisely, type A and type B hexagons are defined by their vertices
\begin{align*}
(a,b) &+ \{(0,0), (0.5, 0.5), (1, 1), (1, 1.5), (0.5, 1), (0, 0.5)\} \quad \text{and} \\
(a,b) &+ \{(0, 0), (- 0.5, -0.5), (-0.5, -1), (0, -0.5), (0.5, 0), (0.5, 0.5)\}
\end{align*}
for $(a, b) \in \mathbb{Z}^2$, respectively, where 
$(a,b) + U := \{(a+x, b+y) \mid (x,y) \in U\}$ for $U \subseteq \mathbb{Z}^2$. 
Then, one can tile a 2D plane with these hexagons. Here, each hexagon represents a central element (a qubit or a resonator) attached to the qubits placed on its vertices via local couplers. 
Note that the types of resonators will be useful for describing the syndrome-extraction circuit for the barbell code. This hexagonal tiling defines the six-qubit star lattice. 
We finish the construction of the architecture by placing additional \nlcs{} connecting each $X$-check qubit at $(a, b)$ to the $Z$-check qubit at $(a - 1.5, b + 2.5)$.
For each pair of $X$-check qubit at $(a,b)$ and $Z$-check qubit at $(a-1.5, b+2.5)$, consider the six hexagons containing the two vertices where the two check qubits are located. Note that these hexagons with the \nlcs{} connecting the check qubits form the shape of a ``barbell''. 
We observe that $X$-type and $Z$-type stabilizers of the weight-8 barbell code corresponding to the qubits at $(a,b)$ and $(a-1.5, b+2.5)$ are supported by the qubits on the six hexagons of the barbell.

\begin{figure}[t]
    \centering
    \includegraphics[width=\linewidth]{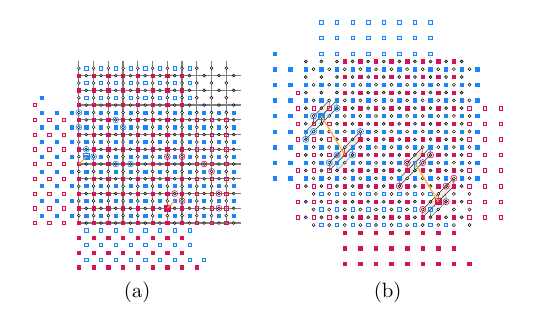}
    \caption{In (a), we show a tile code with an  $X$- and a $Z$-type stabilizer highlighted. For all boundary stabilizers, we added dummy stabilizer qubits. In (b), we show a version of this tile code in which qubits, as well as $X$- and $Z$-type checks, are moved in space. In this way, the support of both the $X$- and the $Z$-type stabilizers is contained in the `barbell' formed by the neighbors of a pair of $X$- and $Z$-checks in the \DHex{} layout. We mention that this is the $[[242,16,8]]$ code for which we run circuit-level noise simulations in the Main Text.}
        \label{fig:main-figure}
\end{figure}

\subsection*{Superdense syndrome extraction circuits of barbell codes}
To measure stabilizers of barbell codes, we adopt the superdense syndrome extraction approach, originally developed to implement color codes on a square grid~\cite{gidney2023newcircuitsopensource}. 
In this section, we provide a full description of the syndrome-extraction cycle for our weight-8 barbell codes. 
We state that the syndrome extraction cycle measures stabilizers of the barbell code; see Supplementary Information for the proof. 

For description, we introduce the following notation. 
By construction, a barbell code has ``barbells'' as building blocks. 
Let $\mathcal{B}$ be the index set of barbells. 
For each barbell $i \in \mathcal{B}$, there are two ancilla qubits, the $X$-check qubit~$a_i^X$ and the $Z$-check qubit $a_i^Z$, corresponding to the $X$- and $Z$-type stabilizers $S_i^X$ and $S_i^Z$.
We saw that both $S_i^X$ and $S_i^Z$ have supports that are subsets of the data qubits in the 6 hexagons adjacent to the two ancilla qubits. 
We label the data qubits in the support of $S_i^X$ and $S_i^Z$ as 
\begin{align}
    \supp(S_i^X) = \{q_{i,1}^X, \dots, q_{i,8}^X\} \qquad \text{and} \qquad \supp(S_i^Z) = \{q_{i,1}^Z, \dots, q_{i,8}^Z\} \,.
\end{align}
Here, some of the data qubits are ``dummy qubits'' for the sake of simplicity since some stabilizers have weights smaller than 8 on the boundary. Later in the description of the syndrome extraction cycle, if there is a gate on a dummy data qubit, one must ignore it.

\begin{table}
    \centering
    \begin{tabular}{c|c||c|c}
        Layer & Circuit & Layer & Circuit \\ \hline \hline
        $-1$ & 
        \makecell[l]{
        \textbf{for} $i \in \mathcal{B}$  \textbf{do} \\ 
        \quad $\sRX \enspace a_i^X$ \\ 
        \quad $\sR \enspace a_i^Z$ \\ 
        \textbf{end for} 
        }  
        & $r \in \{5, 6, 7, 8\}$ & 
        \makecell[l]{
        \textbf{for} $i \in \mathcal{B}$  \textbf{do} \\ 
        \quad $\sCNOT \enspace q_{i,r-4}^Z \enspace a_i^X$  \\
        \quad $\sCNOT \enspace q_{i,r-4}^Z \enspace a_i^Z$  \\
        \textbf{end for} 
        } 
        \\ \hline
        $0$ &  
        \makecell[l]{
        \textbf{for} $i \in \mathcal{B}$  \textbf{do} \\ 
        \quad $\sCNOT \enspace a_i^X \enspace a_i^Z$ \\
        \textbf{end for} 
        }
        & $9$ & 
        \makecell[l]{
        \textbf{for} $i \in \mathcal{B}$  \textbf{do} \\ 
        \quad $\sCNOT \enspace a_i^X \enspace a_i^Z$ \\
        \textbf{end for} 
        }
        \\ \hline
        $r \in \{1, 2, 3, 4\}$ &  
        \makecell[l]{
        \textbf{for} $i \in \mathcal{B}$  \textbf{do} \\ 
        \quad $\sCNOT \enspace a_i^X \enspace q_{i,r}^X$ \\
        \quad $\sCNOT \enspace a_i^Z \enspace q_{i,r}^X$ \\
        \textbf{end for} 
        } 
        & $10$ &
        \makecell[l]{
        \textbf{for} $i \in \mathcal{B}$  \textbf{do} \\ 
        \quad $\sMX \enspace a_i^X$ \\ 
        \quad $\sM \enspace a_i^Z$ \\ 
        \textbf{end for} 
        }  
        \\
    \end{tabular}
    \caption{The syndrome extraction cycle of the weight-8 barbell code. Gates $\sR$ and $\sRX$ reset qubits as $\ket{0}$ and $\ket{+}$. Gates $\sM $ and $\sMX$ measure qubits in $Z$ and $X$ basis.}
    \label{tab:barbell qec cycle}
\end{table}

The syndrome extraction cycle is described in \Cref{tab:barbell qec cycle} and illustrated in \Cref{fig:chip_and_superdense} of the main text. 
In the first two layers, we prepare Bell states on the ancilla qubits $a_i^X$ and $a_i^Z$ for each barbell $i \in \mathcal{B}$ using the \nlc{} connecting $a_i^X$ and $a_i^Z$. 
In the last two layers, we perform Bell measurements on the ancilla qubits using \nlcs{}. 
For each layer $r \in \{1, 2, \dots, w\}$, we apply $\sCNOT$ on the ancilla and data qubits via the couplers in the hexagons. 
We observe two facts from \Cref{fig:chip_and_superdense}. First, for each barbell, if $\sCNOT$s on the $X$-check qubit act on the horizontal (vertical) qubit in a certain round, then $\sCNOT$s on the $Z$-check qubit act on the vertical (horizontal) qubit in that layer, and vice versa. 
Second, for each $ \ sCNOT$, if $\sCNOT$ on the $X$-check qubit is mediated by a type-A (type-B) hexagon in a certain round, then $\sCNOT$ on the $Z$-check is mediated by a type-B (type-A) hexagon in that round, and vice versa. 
From these two facts, we see that the syndrome extraction cycle has circuit depth 12.

Note that there is a freedom of choice for the ordering $\{q_{i,j}^X\}_{j=1}^8$ and $\{q_{i,j}^Z\}_{j=1}^8$. 
The ordering we present was the most performant among 100 randomly chosen orderings in numerical memory experiments.

\subsection*{Numerical simulation}
For all numerical simulations, we consider the uniform depolarizing noise model on the memory and logical Pauli measurement circuits. 
For physical noise strength $p$, we assume that resets and measurements are subject to bit-flip noise of strength $p$, and idling qubits and $\mathsf{CNOT}$s suffer single- and two-qubit depolarizing noise of strength $p$. 
We used Stim~\cite{gidney2021stim} to construct detector error models (DEMs) of our noisy circuits. 
We then used Relay-BP as the decoding algorithm on the DEMs for the Monte Carlo simulation; for details about Relay-BP, see~\cite{mullerImprovedBeliefPropagation2025}. More precisely, we use Relay-BP-5 with a number of legs $R = 300$. As done in~\cite{mullerImprovedBeliefPropagation2025}, we use up to 80 pre-iterations and up to 60 BP-iterations for each leg. 
We decode $X$- and $Z$-type errors separately, disregarding irrelevant detectors. 
For distances 11 and 14, we optimized gamma values separately for lower and higher error rates. We extrapolate the logical error curves by fitting to 
\begin{align}
    \mathrm{LER}(p) = p^{d/2} \exp(c_0 + c_1 p + c_2 p^2)
\end{align}
where $d$ is the distance of the code and $p$ is the physical error rate.

\subsection*{Acknowledgements}
We thank Nikolas P. Breuckmann, Jens Niklas Eberhardt, Jakub Mro\.zek, Frank Deppe, Caspar Ockeloen-Korppi, and Andrew Guthrie for fruitful discussions.

\section*{Supplementary Information}

\section{Six-qubit star lattice plus \nlc{} architecture}
\label{app:6-qubit-star-lattice+nlc}

The six-qubit star lattice plus \nlc{} architecture (6QSL+NLC) inherits most of the properties of the six-qubit star lattice introduced in~\cite{renger2025superconductingqubitresonatorquantumprocessor, vigneau2025quantumerrordetectionstararchitecture,orell2026efficientandaccuratetwoqubitgate}. The unit cell comprises six qubits connected by a multi-mode coupler made of tunable couplers and a central element (\Cref{fig:six-qubit-star}). This topology requires 0.5 central elements and 3 tunable couplers per physical qubit in the bulk. The central elements can be made of qubits or resonators. The cells are organized in a honeycomb lattice forming a \textit{bipartite rhombille} graph~\cite{smithTumblingBlocks2002,conwaySymmetriesThings2016} where qubits and central elements compose the two sets of vertices.

Qubits interact with the central elements via tunable couplers that can be operated to selectively turn their coupling on or off~\cite {mitTunablecoupler2018}.
There are two known protocols to perform controlled-$Z$ (\cz{}) gates between two qubits $q_\mathrm{1}$, $q_\mathrm{2}$ that share a connection to a common central element $q_\mathrm{c}$. 
The first protocol uses \move{}$_{1\mathrm{c}}$ (from $q_1$ to $q_\mathrm{c}$) and \move{}$_{\mathrm{c}1}$ (from $q_\mathrm{c}$ to $q_1$) operations and $\cz{}_{2c}$, a $\cz{}$ gate between $q_2$ and $q_\mathrm{c}$. \move{} operations are similar to an $\mathrm{iSWAP}$ gate, but with the second elements starting in the ground state~\cite{MartinisMOVE, renger2025superconductingqubitresonatorquantumprocessor}.
First \move{}$_{1\mathrm{c}}$ is performed to transfer the state of $q_\mathrm{1}$ into $q_\mathrm{c}$, followed by $\cz{}_{2\mathrm{c}}$ gate between $q_\mathrm{2}$ and $q_\mathrm{c}$, and finally \move{}$_{\mathrm{c}1}$ to transfer back the state from $q_\mathrm{c}$ to $q_\mathrm{1}$. This sequence effectively results in $\cz{} = \move{}_{1\mathrm{c}}-\cz_{2\mathrm{c}}-\move{}_{\mathrm{c}1}$.
The second protocol is a direct \cz{} between $q_\mathrm{1}$ and $q_\mathrm{2}$ using multi-mode interaction carried by $q_\mathrm{c}$ and the two tunable couplers in between~\cite{orell2026efficientandaccuratetwoqubitgate}. Theoretically, this approach allows performing a direct \cz{} gate in $1/\sqrt{2}$ of the time required for a $\move{}_{1\mathrm{c}}-\cz_{2\mathrm{c}}-\move{}_{\mathrm{c}1}$ sequence.
CZ gates with fidelity exceeding 99.4~\% under a total duration of 56~ns have been demonstrated experimentally using both protocols~\cite{Verjauw_APS_talk_2026}.

 \begin{figure}
    \centering
    \includegraphics[width=0.6\linewidth]{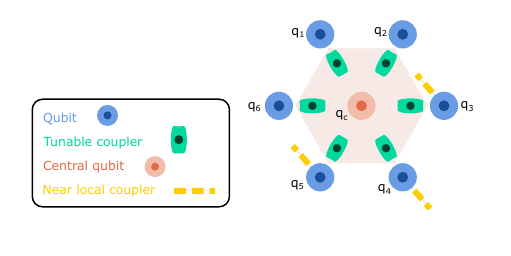}
    \caption{Unit cell of \DHex{} plus \nlc{} (6QSL+NLC) architecture. Six qubits are connected via tunable couplers to a central element in a star topology. The pink hexagon highlights the multi-qubit coupler formed by the six tunable couplers and the central element. The extremities of \nlcs{} are drawn for three qubits (the syndrome qubits of the barbell code).}
    \label{fig:six-qubit-star}
\end{figure}

The barbell codes require, additionally, \nlcs{} to connect pairs of syndrome qubits. Every ancilla qubit (roughly half of the physical qubits) shares one \nlcs{} with another syndrome qubit, making the \nlcs{} count 1/4 per physical qubit. Each physical qubit needs three connections to a tunable coupler and, if it is a syndrome qubit, one connection to a \nlcs{}. Central elements require six connections to tunable couplers but can be operated without control lines or a readout resonator.

\section{Hardware complexity}
\label{app:hardware-complexity}

In this section, we explain a hardware complexity metric~$C_{\mathrm{hw}}$ for superconducting QPU chips, as introduced in~\cite{mathews2026placing}, and compute it for the chip architecture supporting the barbell code. We base our calculation on the weight-8 barbell code. The metric is unchanged for weight-10, since the relevant parameters are identical; however, we expect it to be slightly smaller for weight-6 due to the shorter \nlcs{}.

The metric $C_{\mathrm{hw}}$ takes into account four quantities: the number of tiers ($\mathsf{Tier})$, the average length of the couplers across higher tiers in units of the length of the shortest coupler ($\mathsf{Length})$, the maximum average of bump bonds across all tiers ($\mathsf{Bump}$), and the average number of through-silicon vias per edge on higher tiers ($\mathsf{TSV}$). 
Each quantity $q_i$ for $i \in \{ \mathsf{Tier}, \mathsf{Length}, \mathsf{Bump}, \mathsf{TSV} \}$ is linearly scaled between a \emph{baseline} value $b_i$, reflecting the state-of-the-art hardware implementing surface codes, and an \emph{optimistic} value $p_i$ reflecting the near-future hardware that is believed by the authors of~\cite{mathews2026placing} to be attainable.

The scaled quantity~$c_i$ of $q_i$ is
\begin{equation} \label{eq: ci def}
    c_i=\frac{q_i-b_i}{p_i-b_i} \,,
\end{equation}
and the hardware complexity~$C_{\mathrm{hw}}$ is given by 1 plus the arithmetic mean of the four scaled hardware quantities~$c_i$ with the uniform weight distribution $\{w_i = 1/4\}_i$:
\begin{equation} \label{eq: Chw def}
    C_{\mathrm{hw}}=1+\frac{\sum_i w_i c_i}{\sum_i w_i} \,. 
\end{equation}

By definition, the surface code has the baseline hardware complexity $C_{\mathrm{hw}} = 1$ by plugging $q_i = b_i$. The optimistic hardware complexity is $C_{\mathrm{hw}} = 2$.

Here we provide the four quantities $q_i$ of the chip architecture that supports the barbell code, see \Cref{tab:Chw}.
\begin{itemize}
\item  Since all \nlcs{} are parallel, they can be routed in a single additional tier on top of the base tier. Therefore, the architecture has 2 tiers.
\item  Since there is no crossing of \nlcs{}, the number of bump bonds is 2 for each \nlc{}: one at each end of the non-local coupler.
\item For the same reason, the number of TSVs for each \nlc{} is also 2.
\item  The length of the \nlc{} can be calculated using elementary geometry in~\Cref{fig:chip_and_superdense} as
\begin{equation}
    \ell=\sqrt{( 3.5a )^2 + (13a\sin{(\pi/3)} )^2} \approx 11.8a \,,
\end{equation}
where $a$ is the unit length of the coupler in the base tier. 
Here, we make the simplification that the qubit arrangement perfectly follows a regular hexagonal lattice, but real devices could have variable coupler lengths and qubit positions to accommodate design constraints.

\end{itemize}
From these quantities, we derive $C_{\mathrm{hw}} \approx 1.65$ for the barbell code, mainly dominated by the \nlc{} length. This number is particularly modest in comparison with the other known codes with comparable logical efficiency~\cite{mathews2026placing}, thanks to the small number of tiers, bumps, and TSVs made possible by the fact that all \nlcs{} are parallel. A representation of the distance $8$ barbell code on a physically realistic hexagonal lattice with \nlcs{} is shown in~\Cref{fig:Barbell_DHex}.

Note that we follow the convention of \cite{mathews2026placing} that use one entire tier for the qubit layer, placing the qubit on one face and routing the control lines on the opposite face. 
Therefore, a second tier is required for the \nlcs{}. 
An other option would be to have the qubits and control lines in the same face leaving the opposite face for the \nlcs{}. In this case the all the components would fit in a one single tier and there would not be any TSV. This option would results in an even lower $C_{\mathrm{hw}}$.

We also point out that not all the difficulties of manufacturing QPU for certain codes are captured by the metric. For instance, the number of connections that a qubit needs to support is lower using the \DHex{} architecture, as discussed in Supp.~\Cref{app:6-qubit-star-lattice+nlc}, which is another advantage. On the other hand, parallel \nlcs{} might create crosstalk issues; however, we do not expect this to pose a major obstacle to chip development.

\begin{figure}
    \centering
    \includegraphics[width=\linewidth]{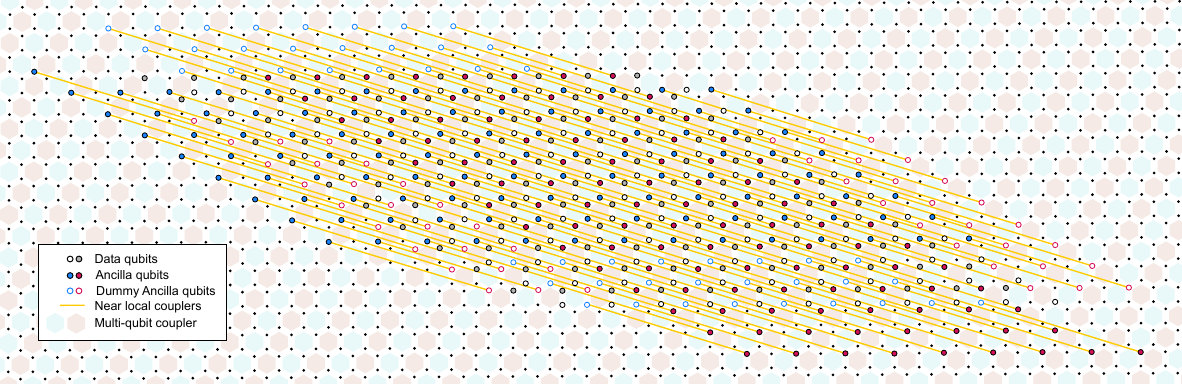}
    \caption{Illustration of the weight-8 barbell code of distance 8 (Fig.~\ref{fig:main-figure}) in \DHex{} plus \nlc{} architecture (6QSL+NLC). For simplicity, the central elements and tunable coupler are only represented in the form of multi-qubit couplers.
   }
    \label{fig:Barbell_DHex}
\end{figure}

\begin{table}[]
    \centering
    \begin{tabular}{|c|c|c|c|c|c|}
        \hline
        Code (quantities)    &  $\mathsf{Tier}$   & $\mathsf{Length}$    & $\mathsf{Bump}$     & $\mathsf{TSV}$   &     $ C_{\mathrm{hw}}$ \\
        \hline \hline
        Baseline (surface), $b_i$ & 1 & 1 & 0 & 0 & 1 \\
        \hline
        Optimistic, $p_i$ & 5 & 10 & 4 & 3 & 2 \\
        \hline
        Barbell, $q_i$ & 2        &  11.8      &   2	   & 2      &  1.65    \\
        \hline
    \end{tabular}
    \caption{Quantities for calculating the hardware complexity~$C_{\mathrm{hw}}$ based on the number of Tiers, the length of the \nlcs{}, the number of Bumps, and of TSVs. For the first and the second row, $C_{\mathrm{hw}}$ is calculated from \cref{eq: ci def,eq: Chw def} by setting $q_i = b_i$ and $q_i = p_i$, respectively.}
    \label{tab:Chw}
\end{table}

\section{Barbell codes -- definition and examples}
In this section, we start with a brief overview of tile codes. Then we describe a construction of a barbell code from a given tile code associated with a connectivity graph imposed by hardware constraints. We also provide specific examples of barbell codes, including the code with weight-8 bulk stabilizers in the main text, as well as codes with weights 6 and 10.

\subsection{Tile codes}
Barbell codes are constructed from tile codes by adding ancilla qubits, translating qubits, and embedding qubits in a proper connectivity graph. 
For a precise definition of a barbell code, we first briefly describe the construction of a tile code. For details, see \cite{steffan2025tile,liang2025planarldpccodes,breuckmann2025logicaloperatorsderivedautomorphisms}. 
As a first step, we choose the shape of the bulk stabilizers so that the minimal rectangular lattice containing them has size $(D+1) \times (D+1)$. 
Data qubits of a tile code sit on horizontal or vertical edges of a lattice. 
In this work, we choose a rectangular lattice of size $L \times M$ as the set of edges where data qubits are located (see \cite{steffan2025tile} for different shapes of lattices, e.g., rotated rectangles)
Explicitly, we define 
\begin{align}
    U_{\mathtt{h}} &= \{(u_1, u_2) \in \mathbb{Z}^2 \mid 0 \leq u_1 \leq L-1, 0 \leq u_2 \leq M - 1 \} \,, \\
    U_{\mathtt{v}} &= \{(u_1, u_2) \in \mathbb{Z}^2 \mid 0 \leq u_1 \leq L-1, 0 \leq u_2 \leq M - 1 \} \,.
\end{align}
Denote $q(\texttt{h}, \mathbf{u})$ and $q(\mathtt{v}, \mathbf{u})$ the data qubits on the horizontal and vertical edges $\{(u_1, u_2), (u_1 + 1, u_2)\}$ and $\{(u_1,u_2), (u_1, u_2+1) \}$, respectively, with $\mathbf{u} = (u_1, u_2)$. 
The sets of the horizontal and the vertical qubits are 
\begin{align}
    \{q(\texttt{h}, \mathbf{u}) \mid \mathbf{u} \in U_\texttt{h} \} \quad \text{and} \quad \{q(\texttt{v}, \mathbf{u}) \mid \mathbf{u} \in U_\texttt{v} \} \,, \quad U_\texttt{h}, U_\texttt{v} \subset \mathbb{Z}^2 \,. 
\end{align}
We put $X$ and $Z$ stabilizers of the tile code at the vertices in the sets~$U_{\mathtt{X}, \mathrm{tile}}$ and $U_{\mathtt{Z}, \mathrm{tile}}$. 
These sets have disjoint subsets $U_{\mathtt{X}, \mathrm{tile,base}}$, $U_{\mathtt{X}, \mathrm{tile,corner}}$ and $U_{\mathtt{Z}, \mathrm{tile,base}}$, $U_{\mathtt{Z}, \mathrm{tile,corner}}$ which are
\begin{align}
    U_{\mathtt{X}, \mathrm{tile,base}} &= \{(u_1, u_2) \in \mathbb{Z}^2 \mid 0 \leq u_1 \leq L-D-1, -D \leq u_2 \leq M-1 \} \,, \\
    U_{\mathtt{Z}, \mathrm{tile,base}} &= \{(u_1, u_2) \in \mathbb{Z}^2 \mid -D \leq u_1 \leq L-1, 0 \leq u_2 \leq M-D-1 \} \,, \\
    U_{\mathtt{X}, \mathrm{tile,corner}} &= U_{\mathtt{X}, \mathrm{tile}} \setminus U_{\mathtt{X}, \mathrm{tile, base}} \,, \\
    U_{\mathtt{Z}, \mathrm{tile,corner}} &= U_{\mathtt{Z}, \mathrm{tile}} \setminus U_{\mathtt{Z}, \mathrm{tile, base}}  \,.
\end{align}
The stabilizers that are not located in these subsets are called the \emph{corner stabilizers}. 
In this work, we can specify all tile codes by the bulk stabilizer shape (including $D$), $L$, $M$, $U_{\mathtt{X}, \mathrm{tile,corner}}$ and $U_{\mathtt{Z}, \mathrm{tile,corner}}$.

\subsection{Definition of barbell codes}\label{sec: def barbell codes}
In the first step of barbell code construction, we add $X$- and $Z$-type ``dummy'' stabilizers as follows. 
We choose an ``offset'' $\mathbf{v}_{\mathtt{XZ}} \in \mathbb{Z}^2$, and for each $\mathbf{u} \in U_{\mathtt{X},\mathrm{tile}}$, if $\mathbf{u}':=\mathbf{u} + \mathbf{v}_{\mathtt{XZ}} \not \in U_{\mathtt{Z},\mathrm{tile}}$, add a dummy $Z$-type stabilizer at $\mathbf{u}'$ whose support is empty. Analogously, for each $\mathbf{u} \in U_{\mathtt{Z},\mathrm{tile}}$, if $\mathbf{u}'' := \mathbf{u}-\mathbf{v}_{\mathtt{XZ}} \not \in U_{\mathtt{X}, \mathrm{tile}}$, then add a dummy $X$-type stabilizer at $\mathbf{u}''$ whose support is empty. 
Denote~$U_{\mathtt{X}}$ and $U_\mathtt{Z}$ the resulting sets of vertices occupied by $X$- and $Z$-type stabilizers after including dummies. 
Explicitly, we have
\begin{align}
    U_{\mathtt{X}} = U_{\mathtt{X}, \mathrm{tile}} \cup \{\mathbf{u} - \mathbf{v}_{\mathtt{XZ}} \mid \mathbf{u} \in U_{\mathtt{Z}, \mathrm{tile}}, \mathbf{u} - \mathbf{v}_{\mathtt{XZ}} \not \in U_{\mathtt{X}, \mathrm{tile}} \} \,, \\
    U_{\mathtt{Z}} = U_{\mathtt{Z}, \mathrm{tile}} \cup \{\mathbf{u} + \mathbf{v}_{\mathtt{XZ}} \mid \mathbf{u} \in U_{\mathtt{X}, \mathrm{tile}}, \mathbf{u} + \mathbf{v}_{\mathtt{XZ}} \not \in U_{\mathtt{Z}, \mathrm{tile}} \} \,.
\end{align}
By construction, we have a one-to-one correspondence 
\begin{align} \label{eq: one-to-one $X$ and $Z$ check qubits}
    U_{\mathtt{X}} \ni \mathbf{u} \mapsto \mathbf{u}+\mathbf{v}_{\mathtt{XZ}} \in U_{\mathtt{Z}} 
\end{align}
between $U_{\mathtt{X}}$ and $U_{\mathtt{Z}} = U_{\mathtt{X}}+ \mathbf{v}_{\mathtt{XZ}}$, where 
$A + \mathbf{v} := \{\mathbf{a} + \mathbf{v} \mid \mathbf{a} \in A\}$ for $A \subseteq \mathbb{Z}^2, \mathbf{v} \in \mathbb{Z}^2$.

In the second step, we place ancilla qubits on the vertices and on the plaquettes that are occupied by stabilizers, including dummies. 
Precisely, we add an ancilla \emph{$X$-check qubit} $q(\mathtt{X}, \mathbf{u})$ at the coordinate $(u_1, u_2)$ for each $\mathbf{u} = (u_1, u_2) \in U_{\mathtt{X}}$.
Similarly, we add an ancilla \emph{$Z$-check qubit}~$q(\mathtt{Z}, \mathbf{u})$ at the coordinate $(u_1 + 0.5, u_2 + 0.5)$ for each $\mathbf{u} = (u_1, u_2) \in U_{\mathtt{Z}}$. 

The sets of $X$ and $Z$ check qubits are 
\begin{align}
    \{q(\texttt{X}, \mathbf{u}) \mid \mathbf{u} \in U_\texttt{X} \} \quad \text{and} \quad \{q(\texttt{Z}, \mathbf{u}) \mid \mathbf{u} \in U_\texttt{Z} \} \,, \quad U_\texttt{X}, U_\texttt{Z} \subset \mathbb{Z}^2 \,. 
\end{align}

In the final step, we translate qubits $\{q(t, \mathbf{u}) \mid \mathbf{u} \in U_t\}$ for $t \in \{\mathtt{X}, \mathtt{Z}, \mathtt{h}, \mathtt{v}\}$ by $\mathbf{v}_t \in \mathbb{Z}^2$. 
Vectors~$\{\mathbf{v}_t\}_{t \in \{\mathtt{X}, \mathtt{Z}, \mathtt{h}, \mathtt{v}\}}$ are chosen so that the whole construction satisfies certain conditions related to the qubit connectivity, which is to be explained at the end of this subsection. 
In the end, the sets of all four types of qubits can be written as
\begin{align}
    \{q(t, \mathbf{u}) \mid \mathbf{u} \in U_t + \mathbf{v}_t \}\,, \mathbf{v}_t \in \mathbb{Z}^2, t \in \{\texttt{X}, \texttt{Z}, \texttt{h}, \texttt{v}\} \,.
\end{align}

Without loss of generality, we do not translate $X$-check qubits, i.e., $\mathbf{v}_{\texttt{X}} = 0$. 
We consistently translate stabilizers as well. 
That is, for a $t$-type ($t \in \{\mathtt{X}, \mathtt{Z}\}$) stabilizer~$S$ of a tile code at $\mathbf{u}$ with its support 
\begin{align}
\supp(S) = \{q(t_i, \mathbf{u}_i') \mid i = 1, \dots, w\} \,, \quad t_1, \dots, t_w \in \{\mathtt{h}, \mathtt{v}\} \,,
\end{align}
the translated $t$-type stabilizer~$S'$ corresponds to the translated check qubit $q(t, \mathbf{u} + \mathbf{v}_t)$ and has support
\begin{align}
\supp(S') = \{q(t_i, \mathbf{u}_i' + \mathbf{v}_{t_i}) \mid i = 1, \dots, w\} \,.
\end{align}

We will define a barbell code based on a ``connectivity graph''~$G = (V, E)$ which reflects hardware constraints restricting the set of available two-qubit gates. 
All two-qubit gates for measuring stabilizers of a barbell code act on pairs of qubits that are adjacent in its connectivity graph. 
The connectivity graph $G$ has vertices that are the data and syndrome qubits, i.e., 
\begin{align}
    V = \bigcup_{t \in \{\mathtt{X}, \mathtt{Z}, \mathtt{h}, \mathtt{v}\}} \{q(t, \mathbf{u}) \mid \mathbf{u} \in U_{t} + \mathbf{v}_t\} \,,
\end{align}
and its edges are constructed as follows. 
We put an edge for the pair of check qubits~$q(\mathtt{X}, \mathbf{u})$ and~$q(\mathtt{Z}, \mathbf{u}+ \mathbf{v}_{\mathtt{XZ}} + \mathbf{v}_{\mathtt{Z}})$ for each $\mathbf{u} \in U_{\mathtt{X}}$. 
Note that these qubits are paired via the correspondence in \cref{eq: one-to-one $X$ and $Z$ check qubits} before the translation by $\mathbf{v}_{\mathtt{Z}}$. 
This edge can be written as 
\begin{align}
    \{q(\texttt{X}, \mathbf{u}), q(\texttt{Z}, \mathbf{u} + \mathbf{v}_{\mathtt{XZ}} + \mathbf{v}_{\texttt{Z}})\} \,, \mathbf{u} \in U_{\texttt{X}}.
\end{align}
Next, we construct the set~$E_{\mathtt{X},\mathbf{u}}$ of edges between the $X$-check qubit $q(\mathtt{X}, \mathbf{u})$ for each $\mathbf{u} \in U_{\mathtt{X}}$ and data qubits using the ``local neighbors'' $N_{\texttt{X}, \texttt{h}}$, $N_{\texttt{X}, \texttt{v}} \subset \mathbb{Z}^2$ of $q(\mathtt{X}, (0,0))$:

\begin{align}
\begin{aligned}
E_{\mathtt{X}, \mathbf{u}}&= E_{\mathtt{X}, \mathbf{u}, \mathbf{h}} \cup E_{\mathtt{X}, \mathbf{u}, \mathbf{v}} \,, \quad \text{where}\\
E_{\mathtt{X}, \mathbf{u}, \mathtt{h}} &= 
    \{\{q(\texttt{X}, \mathbf{u}), q(\texttt{h}, \mathbf{u} + \mathbf{v}'_{\texttt{h}} )\} \mid \mathbf{v}'_{\texttt{h}} \in N_{\texttt{X,h}} \ \text{ if } \ \mathbf{u} + \mathbf{v}'_{\texttt{h}} \in U_{\texttt{h}} + \mathbf{v}_{\texttt{h}} \} \,, \\
E_{\mathtt{X}, \mathbf{u}, \mathtt{v}} &= 
    \{\{q(\texttt{X}, \mathbf{u}), q(\texttt{v}, \mathbf{u} + \mathbf{v}'_{\texttt{v}} )\} \mid \mathbf{v}'_{\texttt{v}} \in N_{\texttt{X,v}} \ \text{ if } \ \mathbf{u} + \mathbf{v}'_{\texttt{v}} \in U_{\texttt{v}} + \mathbf{v}_{\texttt{v}} \} \,.
\end{aligned}
\end{align}
Note that $N_{\mathtt{X}, \mathtt{h}}$ and $N_{\mathtt{X}, \mathtt{v}}$ are the sets of vertices where data qubits are connected to $q(\mathtt{X}, (0,0))$ as their names suggest.
Analogously, we represent the set~$E_{\mathtt{Z}, \mathbf{u}'}$ of edges connecting the $Z$-check qubit~$q(\mathtt{Z},\mathbf{u}')$ with $\mathbf{u}' := \mathbf{u} + \mathbf{v}_{\mathtt{XZ}} + \mathbf{v}_{\mathtt{Z}} \in U_{\mathtt{Z}} + \mathbf{v}_{\mathtt{Z}}$ to data qubits using the local neighbors $N_{\texttt{Z},\texttt{h}}$, $N_{\texttt{Z},\texttt{v}} \subset \mathbb{Z}^2$ of $q(\mathtt{Z},(0,0))$:
\begin{align}
\begin{aligned} E_{\mathtt{Z}, \mathbf{u}'} &= E_{\mathtt{Z}, \mathbf{u}', \mathtt{h}} \cup E_{\mathtt{Z}, \mathbf{u}', \mathtt{v}}\,, \quad \text{where} \\
E_{\mathtt{Z}, \mathbf{u}',\mathtt{h}} &=
    \{\{q(\texttt{Z}, \mathbf{u}'), q(\texttt{h}, \mathbf{u}' + \mathbf{v}'_{\texttt{h}} )\} \mid \mathbf{v}'_{\texttt{h}} \in N_{\texttt{Z,h}} \ \text{ if } \ \mathbf{u}' + \mathbf{v}'_{\texttt{h}} \in U_{\texttt{h}} + \mathbf{v}_{\texttt{h}} \} \,, \\
E_{\mathtt{Z}, \mathbf{u}', \mathtt{v}} &=    \{\{q(\texttt{Z}, \mathbf{u}'), q(\texttt{v}, \mathbf{u}' + \mathbf{v}'_{\texttt{v}} )\} \mid \mathbf{v}'_{\texttt{v}} \in N_{\texttt{Z,v}} \ \text{ if } \ \mathbf{u}' + \mathbf{v}'_{\texttt{v}} \in U_{\texttt{v}} + \mathbf{v}_{\texttt{v}} \} \, .
\end{aligned}
\end{align} 
The set~$E$ of edges of $G$ is then written as 
\begin{align}
    E = \bigcup_{\mathbf{u} \in U_\mathtt{X}} \{ \{q(\texttt{X}, \mathbf{u}), q(\texttt{Z}, \mathbf{u} + \mathbf{v}_{\mathtt{XZ}} + \mathbf{v}_{\texttt{Z}})\}  \} \cup E_{\mathtt{X, \mathbf{u}}} \cup E_{\mathtt{Z}, \mathbf{u + \mathbf{v}_{\mathtt{XZ}} + \mathbf{v}_\mathtt{Z}}} \,.
\end{align}

We call the resulting code with check qubits a barbell code with respect to the connectivity graph $G$ if the following two conditions hold:
\begin{enumerate}
    \item For each $X$-check qubit $q(\texttt{X}, \mathbf{u})$, $\mathbf{u} \in U_\mathtt{X}$, every qubit in the support of the $X$-type stabilizer corresponding to $q(\texttt{X}, \mathbf{u})$ is either adjacent to $q(\texttt{X}, \mathbf{u})$ or $q(\texttt{Z}, \mathbf{u} + \mathbf{u}_{\mathtt{XZ}} + \mathbf{v}_{\texttt{Z}})$ with respect to $G$.
    \item For each $Z$-check qubit $q(\texttt{Z}, \mathbf{u} + \mathbf{v}_{\mathtt{XZ}} + \mathbf{v}_{\texttt{Z}})$, $\mathbf{u} \in U_\mathtt{X}$, every qubit in the support of the $Z$-type stabilizer corresponding to $q(\mathtt{Z}, \mathbf{u} + \mathbf{v}_{\mathtt{XZ}} + \mathbf{v}_{\mathtt{Z}})$ is either adjacent to $q(\texttt{X}, \mathbf{u})$ or $q(\texttt{Z}, \mathbf{u} + \mathbf{v}_{\texttt{Z}})$ with respect to $G$.
\end{enumerate}
By construction, any barbell code with the associated connectivity graph~$G$ is specified by the offset vector~$\mathbf{v}_{\mathtt{XZ}}$, the translation vectors~$\mathbf{v}_{\mathtt{Z}}, \mathbf{v}_{\mathtt{h}}, \mathbf{v}_{\mathtt{v}}$, and the local neighbors $N_{\mathtt{X}, \mathtt{h}}, N_{\mathtt{X}, \mathtt{v}}, N_{\mathtt{Z}, \mathtt{h}}, N_{\mathtt{Z}, \mathtt{v}}$ together with the original tile code.

We observe that the connectivity graph~$G$ is built from the subgraphs, the ``barbells''. 
We formally define a barbell $G_{\mathbf{u}} = (V_{\mathbf{u}}, E_{\mathbf{u}})$ for each $\mathbf{u} \in U_{\mathtt{X}}$ as a subgraph whose sets of vertices and edges are
\begin{align}
    V_{\mathbf{u}} &= \{ q(\mathtt{X}, \mathbf{u}), q(\mathtt{Z}, \mathbf{u} + \mathbf{v}_{\mathtt{XZ}} + \mathbf{v}_{\mathtt{Z}}) \} \cup \mathcal{N}(q(\mathtt{X}, \mathbf{u})) \cup \mathcal{N}(q(\mathtt{Z}, \mathbf{u} + \mathbf{v}_{\mathtt{XZ}} + \mathbf{v}_{\mathtt{Z}}) \,, \label{eq: barbell V} \\
    E_{\mathbf{u}} &= E_{\mathtt{X}, \mathbf{u}} \cup E_{\mathtt{Z}, \mathbf{u}} \,, \label{eq: barbell E}
\end{align}
where
\begin{align}
    \mathcal{N}(q) &= \{q' \in V \mid \text{$q'$ is a data qubit adjacent to $q$ in $G$}\} \quad \text{for each ancilla qubit $q$.}
\end{align}

\subsection{Examples of barbell codes}
In this work, we focus on barbell codes with connectivity graphs that can be realized by the six-qubit star lattice plus \nlc{} (6QSL+NLC) architecture. 
In this architecture, \nlcs{} connect pairs of $X$-check qubits and $Z$-check qubits that are displaced by $\mathbf{v}_{\mathtt{XZ}} + \mathbf{v}_\texttt{Z} + (0.5, 0.5)$. 
The six-qubit star lattice is a tiling of a 2D plane with hexagons where qubits sit on the vertices. It allows ancilla qubits to have local neighbors 
\begin{align}
    N_{\texttt{X}, \texttt{h}} &= \{(-1, -1), (-1, 0), (0,0), (0,1) \}\,, \\ N_{\texttt{X}, \texttt{v}} &= \{(-1, -1), (0, -1), (0,0), (1, 1)\}\,, \\
    N_{\texttt{Z}, \texttt{h}} &= \{(-1, -1), (0, 0), (0, 1), (1, 1) \}\,, \\
    N_{\texttt{Z}, \texttt{v}} &= \{(0, -1), (0, 0), (1,0), (1, 1)\}\,.
\end{align} 
We listed several families of barbell codes on the 6QSL+NLC architecture in \Cref{tab:barbellcodetiles}.

\begin{table}[]
\centering
\renewcommand{\arraystretch}{1.2}

\begin{tabular}{
    C{2.2cm}|
    C{2.7cm}|
    C{6.8cm}|
    C{3.5cm}
}

Tile & Barbell &
$D$, $U_{\mathtt{X}, \mathrm{tile,corner}}$,
$U_{\mathtt{Z},\mathrm{tile,corner}}$,
$\mathbf{v}_{t}$
& $[[n,k,d]]$
\\

\hline\hline

\includegraphics[width=1.5cm]{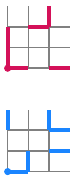}
&
\includegraphics[width=2cm]{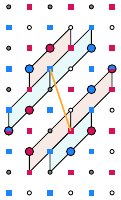}
&
\makecell{
$D = 2$, \\
$U_{\mathtt{X}, \mathrm{tile,corner}} = \emptyset$, \\
$U_{\mathtt{Z}, \mathrm{tile,corner}} = \{(-D, -1)\}$, \\
$\mathbf{v}_{\mathtt{XZ}} = (0,0)$, 
$\mathbf{v}_{\mathtt{Z}} = (-1,1)$ \\
$\mathbf{v}_{\mathtt{h}} = (-2,0)$, 
$\mathbf{v}_{\mathtt{v}} = (-1,-1)$
}
&
\makecell{
$[[128, 7, 6]]$ \\ 
$[[162, 7, 7]]$ \\ 
$[[200, 7, 8]]$ \\ 
$[[242, 7, 9]]$ \\ 
$[[288, 7, 11]]$ \\
$[[338, 7, 13]]$ \\
$[[392, 7, 14]]$
}
\\

\hline

\includegraphics[width=1.5cm]{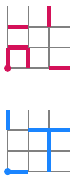}
&
\includegraphics[width=2cm]{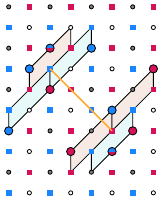}
&
\makecell{
$D = 2$, \\
$U_{\mathtt{X}, \mathrm{tile,corner}} = \emptyset$, \\
$U_{\mathtt{Z}, \mathrm{tile,corner}} = \emptyset$, \\
$\mathbf{v}_{\mathtt{XZ}} = (-1,0)$, 
$\mathbf{v}_{\mathtt{Z}} = (-1,1)$ \\
$\mathbf{v}_{\mathtt{h}} = (-2,0)$, 
$\mathbf{v}_{\mathtt{v}} = (-1,-1)$
}
&
\makecell{
$[[128, 8, 6]]$ \\
$[[162, 8, 7]]$ \\ 
$[[200, 8, 9]]$ \\ 
$[[242, 8, 10]]$ \\ 
$[[288, 8, 12]]$ \\
$[[338, 8, 13]]$ \\
$[[392, 8, 15]]$
}
\\

\hline

\includegraphics[width=1.5cm]{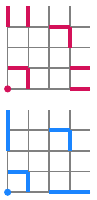}
&
\includegraphics[width=2cm]{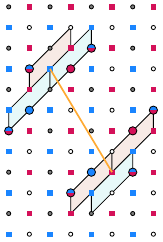}
&
\makecell{
$D = 3$, \\
$U_{\mathtt{X}, \mathrm{tile,corner}} = \{(L-D, -D)\}$, \\
$U_{\mathtt{Z}, \mathrm{tile,corner}} = \{(-D, M-D)\}$, \\
$\mathbf{v}_{\mathtt{XZ}} = (0,0)$, 
$\mathbf{v}_{\mathtt{Z}} = (-2,2)$ \\
$\mathbf{v}_{\mathtt{h}} = (-3,0)$, 
$\mathbf{v}_{\mathtt{v}} = (-2,-1)$
}
&
\makecell{
$[[242, 16, 8]]$ \\ 
$[[288, 16, 9]]$ \\
$[[338, 16, 10]]$ \\
$[[392, 16, 11]]$ \\
$[[450, 16, 14]]$
}
\\

\hline

\includegraphics[width=1.5cm]{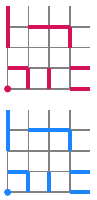}
&
\includegraphics[width=2cm]{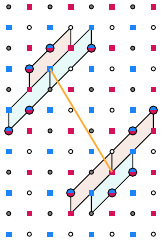}
&
\makecell{
$D = 3$, \\
$U_{\mathtt{X}, \mathrm{tile,corner}} = \{(L-D, -D)\}$, \\
$U_{\mathtt{Z}, \mathrm{tile,corner}} = \{(-D, M-D)\}$, \\
$\mathbf{v}_{\mathtt{XZ}} = (0,0)$, 
$\mathbf{v}_{\mathtt{Z}} = (-2,2)$ \\
$\mathbf{v}_{\mathtt{h}} = (-3,0)$, 
$\mathbf{v}_{\mathtt{v}} = (-2,-1)$
}
&
\makecell{
$[[162, 16, 8]]$ \\
$[[200, 16, 10]]$ \\
$[[242, 16, 11]]$ \\ 
$[[288, 16, 12]]$ \\
$[[338, 16, 13]]$ \\
$[[392, 16, 14]]$ \\
$[[450, 16, 15]]$ \\ 
$[[512, 16, 16]]$
}
\\

\end{tabular}

\caption{
Examples of barbell codes on the six-qubit star lattice + \nlc{} architecture.
The number $n$ of qubits of each code is determined as
$n = 2(L+D)(M+D)$ with $L = M$.
}

\label{tab:barbellcodetiles}

\end{table}
\section{Superdense syndrome extraction cycle}\label{sec: superdense}
In this section, we construct the superdense syndrome-extraction cycle, or simply the QEC cycle, for barbell codes and prove that it measures the parity of the barbell code stabilizers. 

\begin{table}
    \centering
    \begin{tabular}{c|c||c|c}
        Layer & Circuit & Layer & Circuit \\ \hline \hline
        $-1$ & 
        \makecell[l]{
        \textbf{for} $i \in \mathcal{B}$  \textbf{do} \\ 
        \quad $\sRX \enspace a_i^X$ \\ 
        \quad $\sR \enspace a_i^Z$ \\ 
        \textbf{end for} 
        }  
        & $r \in \{w/2 + 1,  w/2 + 2, \dots, w\}$ & 
        \makecell[l]{
        \textbf{for} $i \in \mathcal{B}$  \textbf{do} \\ 
        \quad $\sCNOT \enspace q_{i,r-w/2}^Z \enspace a_i^X$  \\
        \quad $\sCNOT \enspace q_{i, r}^Z \enspace a_i^Z$  \\
        \textbf{end for} 
        } 
        \\ \hline
        $0$ &  
        \makecell[l]{
        \textbf{for} $i \in \mathcal{B}$  \textbf{do} \\ 
        \quad $\sCNOT \enspace a_i^X \enspace a_i^Z$ \\
        \textbf{end for} 
        }
        & $w+1$ & 
        \makecell[l]{
        \textbf{for} $i \in \mathcal{B}$  \textbf{do} \\ 
        \quad $\sCNOT \enspace a_i^X \enspace a_i^Z$ \\
        \textbf{end for} 
        }
        \\ \hline
        $r \in \{1, 2, \dots, w/2\}$ &  
        \makecell[l]{
        \textbf{for} $i \in \mathcal{B}$  \textbf{do} \\ 
        \quad $\sCNOT \enspace a_i^X \enspace q_{i,r}^X$ \\
        \quad $\sCNOT \enspace a_i^Z \enspace q_{i,r + w/2}^X$ \\
        \textbf{end for} 
        } 
        & $w+2$ &
        \makecell[l]{
        \textbf{for} $i \in \mathcal{B}$  \textbf{do} \\ 
        \quad $\sMX \enspace a_i^X$ \\ 
        \quad $\sM \enspace a_i^Z$ \\ 
        \textbf{end for} 
        }  
        \\
    \end{tabular}
    \caption{The syndrome extraction cycle of the weight-$w$ barbell code. Gates $\sR$ and $\sRX$ reset qubits as $\ket{0}$ and $\ket{+}$. Gates $\sM $ and $\sMX$ measure qubits in $Z$ and $X$ basis.}
    \label{tab:barbell general qec cycle}
\end{table}

Consider a barbell code based on a tile code in which the bulk stabilizers have weight $w$. 
We denote by
\begin{align}
    \mathcal{B} = \{G_{\mathbf{u}} = (V_{\mathbf{u}}, E_{\mathbf{u}}) \mid \mathbf{u} \in U_{\mathtt{X}}\}
\end{align}
the set of barbells in the weight-$w$ barbell code; see \Cref{sec: def barbell codes} for the precise definition of barbells.
For each barbell $i = G_{\mathbf{u}} \in \mathcal{B}$, denote 
\begin{align}
    a_i^X = q(\mathtt{X}, \mathbf{u}) \,, \quad a_i^Z = q(\mathtt{Z}, \mathbf{u} + \mathbf{v}_{\mathtt{XZ}} + \mathbf{v}_{\mathtt{Z}})
\end{align} 
the ancilla $X$- and $Z$-check qubits. 
We also denote $S_i^X$ and $S_i^Z$ as the corresponding $X$- and $Z$-type stabilizers of the barbell code. 
We label the qubits in the support of these stabilizers as 
\begin{align}
    \supp(S_i^X) = \{q_{i,1}^X, \dots, q_{i,w}^X\} \qquad \text{and} \qquad \supp(S_i^Z) = \{q_{i,1}^Z, \dots, q_{i,w}^Z\} \,.
\end{align} 
Note that the supports of $S_i^X$ and $S_i^Z$ are partitioned into subsets of $\mathcal{N}(a_i^X)$ and $\mathcal{N}(a_i^Z)$, which are the subsets of data qubits adjacent to the ancilla check qubits $a_i^X$ and $a_i^Z$.
For the sake of simplicity, we assume that $w$ is even and the supports of both $S_i^X$ and $S_i^Z$ are partitioned into subsets of $\mathcal{N}(a_i^X)$ and $\mathcal{N}(a_i^Z)$ consisting of $w/2$ qubits. 
We also assume that qubits $(q_i^X)_{i=1}^w$ and $(q_i^Z)_{i=1}^w$ are ordered so that 
\begin{align}
    \mathcal{N}(a_i^X) \cap \supp(S_i^X) = \{q_1^X, \dots, q_{w/2}^X\} \,, \quad
    \mathcal{N}(a_i^Z) \cap \supp(S_i^X) = \{q_{w/2+1}^X, \dots, q_{w}^X\} \,, \\
    \mathcal{N}(a_i^X) \cap \supp(S_i^Z) = \{q_1^Z, \dots, q_{w/2}^Z\} \,, \quad
    \mathcal{N}(a_i^Z) \cap \supp(S_i^Z) = \{q_{w/2+1}^Z, \dots, q_{w}^Z\} \,.
\end{align}

We define the superdense syndrome extraction cycle of a barbell code as the circuit described in \Cref{tab:barbell general qec cycle}. 
In the first two layers (Layer $-1$ and $0$), a Bell pair on the syndrome qubits $a_i^X$ and $a_i^Z$ is prepared for each barbell $i \in \mathcal{B}$. Then, in the following $w/2$ layers (Layer $1$ to $w/2$), $\sCNOT$ gates are applied where control qubits are syndrome qubits and target qubits are data qubits. In the next $w/2$ layers (Layer $w/2 + 1$ to Layer $w$), $\sCNOT$ gates are again applied, but here, control qubits are data qubits and target qubits are syndrome qubits. Finally, in the last two layers (Layer $w+1$ and $w+2$), we perform Bell measurements on the qubits $a_i^X$ and $a_i^Z$ for each $i \in \mathcal{B}$. 
\begin{remark}
    For each layer $r \in \{1, \dots, w/2\}$ of the QEC cycle in \Cref{tab:barbell general qec cycle}, the ordering of $\sCNOT$ gates can be arbitrary, since none of the target qubits is used as a control qubit in that layer. Therefore, the QEC cycle is well-defined. 
\end{remark}
\begin{remark}
    In general, each layer of the QEC cycle may consist of multiple layers of $\sCNOT$ gates due to potential overlaps in target qubits. 
    However, one can find proper orderings of $\{q_{i, j}^X\}_{j=1}^w$ and $\{q_{i, j}^Z\}_{j=1}^w$  for all examples of barbell codes in \Cref{tab:barbellcodetiles} such that each layer has exactly one layer of $\sCNOT$ gates. 
\end{remark}

Next, we show that the syndrome extraction cycle measures stabilizers in the following theorem. 
\begin{theorem} \label{thm: correctness qec cycle}
    There exists a Pauli operator~$U_{\mathrm{frame}}$ that is a product of single-qubit Pauli operators controlled by some measurement results of the syndrome extraction cycle in \Cref{tab:barbell general qec cycle} such that the circuit consisting of the syndrome extraction cycle followed by $U_{\mathrm{frame}}$ measures all barbell code stabilizer generators $\{S_i^X\}_{i \in \mathcal{B}}$ and $\{S_i^Z\}_{i \in \mathcal{B}}$ while preserving the logical information. 
    Explicitly, the following statements hold.
    \begin{enumerate}
        \item Through the syndrome extraction cycle followed by the Pauli operator~$U_{\mathrm{frame}}$, the single-qubit $X_{a_i^X}$ evolves into $(-1)^{m_i^X} S_i^X$, and the single-qubit $Z_{a_i^Z}$ evolves into $(-1)^{m_i^Z} S_i^Z$, where $m_i^X$ and $m_i^Z$ are the measurement results on the qubit $a_i^X$ and $a_i^Z$ in the syndrome extraction cycle. 
        \item For each logical Pauli operator~$L$, consider a Pauli subgroup generated by all stabilizers and $L$. 
        Then, the Pauli subgroup evolves into itself via the syndrome-extraction circuit, followed by $U_{\mathrm{frame}}$. 
    \end{enumerate}

\end{theorem} 
Before we prove \Cref{thm: correctness qec cycle}, we briefly explain the role of the Pauli operator~$U_{\mathrm{frame}}$ in \Cref{thm: correctness qec cycle}. 
The Pauli operator~$U_{\mathrm{frame}}$ is used to match the Pauli frame with measurement results. 
Precisely, if we do not apply $U_{\mathrm{frame}}$ after the syndrome extraction cycle, the single-qubit Pauli~$X_{a_i^X}$ evolves into $(-1)^{m_i^X + M}S_i^X$, where $M \in \{0,1\}$ is the parity of the measurement results on some $X$-check qubits.
In the proof of \Cref{thm: correctness qec cycle}, we construct the Pauli~$U_{\mathrm{frame}}$ and show that 
\begin{align}
    U_{\mathrm{frame}} \cdot \left( (-1)^{m_i^X + M} S_i^X \right) \cdot U_{\mathrm{frame}}^\dagger = (-1)^{m_i^X} S_i^X \,,
\end{align}
implying that $\{m_i^X\}_{i \in \mathcal{B}}$ are syndromes of $X$-type stabilizers. 
Note that in practice, we do not need to physically apply the Pauli operator~$U_{\mathrm{frame}}$ because the Pauli frame can be tracked in software.

To prove \Cref{thm: correctness qec cycle}, we need the following lemma and corollary.
\begin{lemma}\label{lem:odd odd overlap}
    Let $P^X$ and $P^Z$ be $X$- and $Z$-type Pauli operators on data qubits and $p, q \in \{0,1\}$. Define
    \begin{align}
      \mathcal{B}^Z_{pq}(P^X) := \{ j \in \mathcal{B} \mid |\mathcal{N}(a_j^X) \cap \supp(S_j^Z) \cap \supp(P^X)|  \equiv p \quad (\modular 2) \quad \text{and} \nonumber\\
      |\mathcal{N}(a_j^Z) \cap \supp(S_j^Z) \cap \supp(P^X)| \equiv q \quad (\modular 2)
      \} \,, \\
      \mathcal{B}^X_{pq}(P^Z) := \{ j \in \mathcal{B} \mid |\mathcal{N}(a_j^X) \cap \supp(S_i^X) \cap \supp(P^Z)|  \equiv p \quad (\modular 2) \quad \text{and} \nonumber\\ 
      |\mathcal{N}(a_j^Z) \cap \supp(S_i^X) \cap \supp(P^Z)| \equiv q \quad (\modular 2)
      \} \,.
    \end{align}
    Then, the following statements hold:
    \begin{enumerate}[label=\roman*)]
        \item If $P^X$ commutes with all $Z$-type stabilizers $S_j^Z$, then $\mathcal{B}$ is partitioned into disjoint subsets $\mathcal{B}_{00}^Z(P^X)$ and $\mathcal{B}_{11}^Z(P^X)$.
        \item If $P^Z$ commutes with all $X$-type stabilizers $S_j^X$, then $\mathcal{B}$ is partitioned into disjoint subsets $\mathcal{B}_{00}^X(P^Z)$ and $\mathcal{B}_{11}^X(P^Z)$.
    \end{enumerate}
\end{lemma}
\begin{proof}
    We prove the statement for $P^X$; the other follows by swapping $X$ and $Z$.
    By definition, $\mathcal{B}$ is partitioned into four disjoint subsets $\mathcal{B}_{pq}^Z$, $p, q \in \{0,1\}$. 
    Using the fact that $P^X$ commutes with $S_j^Z$ for all $j \in \mathcal{B}$, and since the support of $S_j^Z$ is partitioned into the $X$-check part $\mathcal{N}(a_j^X) \cap \supp(S_j^Z)$ and the $Z$-check part~$\mathcal{N}(a_j^Z) \cap \supp(S_j^Z)$, then the overlap of the support of $P^X$ and the $X$-check part of $S_j^Z$ has the same parity as that of the support of $P^X$ and the $Z$-check part of $S_j^Z$. Therefore, $\mathcal{B}_{01}^Z(P^X)$ and $\mathcal{B}_{10}^Z(P^X)$ are empty.
\end{proof}
\begin{corollary}\label{cor: odd odd}
    Let $U_Z$ be the product of $\sCNOT$ gates from Layer $w/2+1$ to $w$ of the QEC cycle in \Cref{tab:barbell general qec cycle}. Then, for any $X$-type Pauli operator~$P^X$ commuting with all $Z$-type stabilizers~$\{S_i^Z\}_{i \in \mathcal{B}}$, we have
    \begin{align}\label{eq: cor to show with PX}
        U_Z P^X {U_Z}^\dagger = P^X \prod_{j \in \mathcal{B}^Z_{11}(P^X)} X_{a_j^X} X_{a_j^Z} \,.
    \end{align}
    Analogously, for any $Z$-type Pauli operator~$P^Z$ commuting with all $X$-type stabilizers~$\{S_i^X\}_{i \in \mathcal{B}}$, we have 
    \begin{align} \label{eq: cor to show with PZ}
        U_{X} P^Z U_{X}^\dagger = P^Z \prod_{j \in \mathcal{B}^X_{11}(P^Z)} Z_{a_j^X} Z_{a_j^Z} \,,
    \end{align}
    where $U_{X}$ is the product of $\sCNOT$ gates from Layer $1$ to $w/2$ in \Cref{tab:barbell qec cycle}.
\end{corollary}
\begin{proof}
    We only prove that \Cref{eq: cor to show with PX} holds, since \Cref{eq: cor to show with PZ} can be proved by the analogous argument with swapping $X$ and $Z$.  
    For each $j \in \mathcal{B}$, $\sCNOT$ gates in $U_{Z}$ acting on the qubits $a_j^X$ and $a_j^Z$ are those having control qubits in the $X$-check and $Z$-check part of $S_j^Z$, respectively, and the support of $S_j^Z$ comprises those control qubits.
    Thus, for any data qubit $q$, 
    \begin{align}\label{eq:Xq propagation}
        U_{Z} X_q U_{Z}^\dagger = 
        X_q 
        \left(
        \prod_{j \in \mathcal{B}: q \in \mathcal{N}(a_j^X) \cap \supp(S_j^Z)} X_{a_j^X}
        \right)
        \left(
        \prod_{j \in \mathcal{B}: q \in \mathcal{N}(a_j^Z) \cap \supp(S_j^Z)} X_{a_j^Z}
        \right)\,.
    \end{align} 
    Decomposing $P^X$ into the product of single-qubit Paulis $\{X_q\}_{q \in \supp(P^X)}$ and using \Cref{eq:Xq propagation} for each $q \in \supp(P^X)$, we have 
    \begin{align}
        U_{Z} P^X U_{Z}^\dagger &= P^X 
        \left( \prod_{ j \in \mathcal{B}_{10}^Z(P^X) \cup \mathcal{B}_{11}^Z(P^X)} X_{a_j^X} \right) \left( \prod_{ j \in \mathcal{B}_{01}^Z(P^X) \cup \mathcal{B}_{11}^Z(P^X)} X_{a_j^Z} \right) \\
        &= P^X \prod_{j \in \mathcal{B}_{11}^Z(P^X)} X_{a_j^X} X_{a_j^Z} \,, \label{eq: corollary SX propagation result}
    \end{align}
    where \Cref{eq: corollary SX propagation result} holds by \Cref{lem:odd odd overlap}.
\end{proof}

\begin{proof}[Proof of Theorem 1]
    To verify the correctness of the QEC cycle, we first need to track the evolution of single-qubit Paulis $X_{a_i^X}$ and $Z_{a_i^Z}$ at the ancilla qubits $a_i^X$ and $a_i^Z$ for each $i \in \mathcal{B}$.

    We observe that the single-qubit Pauli ${X_{a_i^X}}$ evolves as follows:
    \begin{align}
        X_{a_i^X} &\mapsto X_{a_i^X} X_{a_i^Z} \quad \text{(after Layer $0$)} \label{eq: Xa map 1} \\
        &\mapsto S_i^X X_{a_i^X} X_{a_i^Z} \quad \text{(after Layer $w/2$)} \label{eq: Xa map 2} \\ 
        &\mapsto S_i^X X_{a_i^X} X_{a_i^Z} \prod_{j \in \mathcal{B}_{11}^Z(S_i^X)} (X_{a_j^X} X_{a_j^Z}) \quad \text{(after Layer $w$)} \label{eq: Xa map 3} \\ 
        &\mapsto S_i^X X_{a_i^X} \prod_{j \in \mathcal{B}^Z_{11}(S_i^X)} X_{a_j^X} \quad \text{(after Layer $w+1$)} \label{eq: stab Xpi before measurement}
    \end{align}
    Here, \Cref{eq: Xa map 1,eq: Xa map 2,eq: stab Xpi before measurement} clearly hold by the construction of the QEC cycle, and \Cref{eq: Xa map 3} holds by \Cref{cor: odd odd}. 
    Denote $m_i^X \in \{0,1\}$ the outcome of the single-qubit measurement on the qubit $a_i^X$ in Layer~$w+2$.  
    Then, by measurements the stabilizer in \Cref{eq: stab Xpi before measurement} is transformed into
    \begin{align} \label{eq: unfixed stabilizer}
        S_i^X X_{a_i^X} \prod_{j \in \mathcal{B}^Z_{11}(S_i^X)} X_{a_j^X} \mapsto
        (-1)^{m_i^X + \sum_{j \in \mathcal{B}^Z_{11}(S_i^X)} m_j^X} S_i^X \,.
    \end{align}
    We observe that the QEC cycle measures the stabilizer $S_i^X$, and the corresponding syndrome is the parity of 
    \begin{align}
        m_i^X + \sum_{j \in \mathcal{B}^Z_{11}(S_i^X)} m_j^X \,.
    \end{align}
    Instead of tracking this ``Pauli frame'' for each QEC cycle, one can virtually apply the ``Pauli frame correction operator'' $U_{\mathrm{frame}}$ defined as 
    \begin{align}
        U_{\mathrm{frame}} = U_{\mathrm{frame}}[\{m_i^X\}_{i \in \mathcal{B}}] :=& \prod_{k \in \mathcal{B}} U_{\mathrm{frame},k} \,,
    \end{align}
    where
    $
        U_{\mathrm{frame},k} = U_{\mathrm{frame},k}[\{m_i^X\}_{i \in \mathcal{B}}] := \prod_{q \in \mathcal{N}(a_k^X) \cap \supp(S_k^Z)} Z_q^{m_k^X}
    $.
    By the definition of $\mathcal{B}_{11}^Z(S_i^X)$, we have 
    \begin{align}
    U_{\mathrm{frame}, k} \cdot S_i^X \cdot U_{\mathrm{frame},k}^\dagger=
    \begin{cases}
    (-1)^{m_k^X}S_i^X & \text{ if } k \in \mathcal{B}^Z_{11}(S_i^X),\\
    S_i^X & \text{ otherwise. }
    \end{cases}
    \end{align}
    Therefore, by conjugating $S_i^X$ with $U_{\mathrm{frame},k}$ for all $k$ repeatedly, we have
    \begin{align}\label{eq: SX propagation with frame}
        U_{\mathrm{frame}} \cdot \left( (-1)^{m_i^X + \sum_{j \in \mathcal{B}_{11}(S_i^X)}{m^X_j}}S_i^X \right) \cdot U_{\mathrm{frame}}^\dagger = (-1)^{m^X_i} S_i^X \,.
    \end{align}
    That is, the measurement result $m_i^X$ is indeed the syndrome of the stabilizer $S_i^X$ after the Pauli frame correction by $U_{\mathrm{frame}}$.

    Next, we track the transformation of $Z_{a_i^Z}$ over the QEC cycle as we did for $X_{a_i^X}$. 
    Denote $\{m_i^Z\}_{i \in \mathcal{B}}$ the measurement results  of the qubits $\{a_i^Z\}_{i \in \mathcal{B}}$. From the construction of the QEC cycle, one can observe that the single-qubit Pauli $Z_{a_i^Z}$ evolves as follows:
    \begin{align}
        Z_{a_i^Z} &\mapsto Z_{a_i^X} Z_{a_i^Z} \quad \text{(after Layer 0)} \label{eq: Za map 1} \\
        &\mapsto Z_{a_i^X} Z_{a_i^Z} \quad \text{(after Layer $w/2$)} \label{eq: Za map 2} \\ 
        &\mapsto S_i^Z Z_{a_i^X} Z_{a_i^Z} \quad \text{(after Layer $w$)} \label{eq: Za map 3} \\ 
        &\mapsto (-1)^{m_i^Z} S_i^Z \quad \text{(after Layer $w+2$) } \label{eq: stab Zi after measurement}
    \end{align}
    Since $U_{\mathrm{frame}}$ is a Pauli $Z$ operator, we have 
    \begin{align}
        U_{\mathrm{frame}} \cdot \left( (-1)^{m_i^Z} S_i^Z \right) \cdot U_{\mathrm{frame}}^\dagger = (-1)^{m_i^Z} S_i^Z \,.
    \end{align}
    That is, the measurement result $m_i^Z$ is the syndrome of the stabilizer $S_i^Z$, regardless of $U_{\mathrm{frame}}$. 
    
    Next, we prove the statement about the logical operator. 
    For each logical operator $L$, let $\mathcal{S} = \mathcal{S}[L]$ be the stabilizer group generated by $L$ and all stabilizers of the code, and consider the group $\mathcal{S}'$ obtained by evolving all elements of $\mathcal{S}$ through the QEC cycle and the Pauli frame correction. 
    We want to show that $\mathcal{S} = \mathcal{S}'$, and for that, it suffices to check $L \in \mathcal{S}'$. This is because one can also choose $L$ as a stabilizer in $\mathcal{S}$, which is a trivial logical operator. 
    We verify that $L \in \mathcal{S}'$ holds for both logical $X$ and $Z$ operators~$L=L^X$ and $L=L^Z$, which are products of $X$- and $Z$-type single-qubit Paulis, respectively, since barbell codes are quantum CSS codes.
    
    Through the QEC cycle and the Pauli frame correction operator, the evolution of $L^X$ is given by
    \begin{align}
        L^X &\mapsto L^X \quad  \text{(after Layer $w/2$)} \label{eq: LX mapsto 1} \\
        &\mapsto L^X \prod_{j \in \mathcal{B}^Z_{11}(L^X) }X_{a_j^X}X_{a_j^Z} \quad \text{(after Layer $w$)} \label{eq: LX mapsto 2} \\ 
        &\mapsto L^X \prod_{j \in \mathcal{B}^Z_{11}(L^X)} X_{a_j^X} \quad \text{(after Layer $w+1$)} \label{eq: LX mapsto 3} \\
        &\mapsto (-1)^{\sum_{j \in \mathcal{B}^Z_{11}(L^X)} m_j^X} L^X \quad \text{(after Layer $w+2$)} \label{eq: LX mapsto 4} \\ 
        &\mapsto L^X \quad \text{(after $U_{\mathrm{frame}}$)} \label{eq: LX propagation through frame}
    \end{align}
    Here, \Cref{eq: LX mapsto 1,eq: LX mapsto 3,eq: LX mapsto 4} are implied by the construction of the QEC cycle, \Cref{eq: LX mapsto 2} holds by \Cref{cor: odd odd}, and \Cref{eq: LX propagation through frame} holds by \Cref{eq: SX propagation with frame} together with replacing $S_i^X$ with $L^X$. Thus, $L^X \in \mathcal{S}'$.

    Next, we track the transformation of a logical $Z$ operator $L^Z$. We have 
    \begin{align}
        L^Z &\mapsto L^Z \quad \text{(after Layer $0$)} \label{eq: LZ mapsto 0} \\
        &\mapsto L^Z \prod_{j \in \mathcal{B}_{11}^X(L^Z)} Z_{a_j^X} Z_{a_j^X} \quad \text{(after Layer $w/2$).} \label{eq: LZ mapsto}
    \end{align}
    \Cref{eq: LZ mapsto 0} holds by the construction of the QEC cycle, and \Cref{eq: LZ mapsto} follows from \Cref{cor: odd odd}.
    We also observe from the construction of the QEC cycle that $\prod_{j \in \mathcal{B}_{11}^X(L^Z)} Z_{a_j^Z}$ evolves to
    \begin{align}
        \prod_{j \in \mathcal{B}_{11}^X(L^Z)} Z_{a_j^Z} 
        &\mapsto 
        \prod_{j \in \mathcal{B}_{11}^X(L^Z)} Z_{a_j^X} Z_{a_j^Z} \quad \text{(after Layer $0$)} \label{eq: ajZ evolves 0}  \\  
        &\mapsto 
        \prod_{j \in \mathcal{B}_{11}^X(L^Z)} Z_{a_j^X} Z_{a_j^Z} \quad \text{(after Layer $w/2$)} \label{eq: ajZ evolves 1}
    \end{align}
    From \Cref{eq: LZ mapsto,eq: LZ mapsto 0,eq: ajZ evolves 0,eq: ajZ evolves 1}, we see that the product $L^Z \prod_{j \in \mathcal{B}_{11}^X(L^Z)} Z_{a_j^Z}$ gets transformed into  
    \begin{align}
        L^Z \prod_{j \in \mathcal{B}_{11}^X(L^Z)} Z_{a_j^Z} 
        &\mapsto L^Z \quad \text{(after Layer $w/2$)} \label{eq: LZ and ancilla evolution 1} \\ 
        &\mapsto L^Z \quad \text{(after Layer $w+2$)} \label{eq: LZ and ancilla evolution 2}  \\
        &\mapsto L^Z \quad \text{(after $U_{\mathrm{frame}}$)} \label{eq: LZ and ancilla evolution 3} \,,
    \end{align}
    where \Cref{eq: LZ and ancilla evolution 2} and \Cref{eq: LZ and ancilla evolution 3} are implied by the construction of the QEC cycle.
    Consider the stabilizer group $\mathcal{S}_1$ generated by $\mathcal{S}$ and single-qubit Paulis $\{X_{a_j^X}\}_{j \in \mathcal{B}}$ and $\{Z_{a_j^Z}\}_{j \in \mathcal{B}}$. 
    That is, $\mathcal{S}_1$ is obtained by evolving $\mathcal{S}$ through Layer $-1$ of the syndrome extraction cycle. 
    Since $L^Z \in \mathcal{S}_1$ and $Z_{a_j^Z} \in \mathcal{S}_1$ for all $j \in \mathcal{S}_1$, we have $L^Z \prod_{j \in \mathcal{B}_{11}^X(L^Z)} Z_{a_j^Z} \in \mathcal{S}_1$ as $\mathcal{S}_1$ is closed under multiplication.
    Therefore, \Cref{eq: LZ and ancilla evolution 1,eq: LZ and ancilla evolution 2,eq: LZ and ancilla evolution 3} imply that $L^Z \in \mathcal{S}'$.
\end{proof}

\end{document}